\documentclass{article}

\newif\ifdraftversion

\draftversionfalse
\ifdraftversion
\usepackage{showlabels}
\usepackage{lineno}
\linenumbers
\fi

\IfFileExists{MinionPro.sty}
{\usepackage[lf]{MinionPro}
\def\interior{\righthalfcup}
\usepackage{amsmath,amsthm,amsbsy}}
{\usepackage{times}
\usepackage{savesym} 
\usepackage{amsmath,amsthm,amsbsy}
\savesymbol{iint}
\savesymbol{openbox}
\usepackage{txfonts}
\def\interior{\mathbin{\lrcorner}}
\DeclareSymbolFont{symbols}{OMS}{cmsy}{m}{n}
\restoresymbol{TXF}{iint}
\restoresymbol{TXF}{openbox}}

\usepackage{enumerate}
\usepackage{mathtools}
\usepackage{geometry}
\usepackage{color}
\definecolor{labelkey}{rgb}{0,0,.75}
\definecolor{MyGreen}{rgb}{0,.6,.2}
\definecolor{MyDarkBlue}{rgb}{.1,.1,.75}
\usepackage[colorlinks=true, citecolor=MyGreen, linkcolor=MyDarkBlue]{hyperref}
\usepackage{tensor}
\usepackage[all,cmtip]{xy}
\usepackage{graphicx}
\usepackage{mathrsfs}
\usepackage{mathtools}

\DeclareMathOperator{\Ric}{\mathrm{Ric}}
\DeclareMathOperator{\Id}{\mathrm{Id}}
\DeclareMathOperator{\ck}{\bf L}
\renewcommand{\div}{\mathop{\rm div}\nolimits}

\DeclareMathOperator{\tr}{\rm tr}
\DeclareMathOperator{\Lie}{\mathrm{Lie}}

\renewcommand{\Im}{\mathop{\mathrm{Im}}\nolimits}
\renewcommand{\Re}{\mathop{\mathrm{Re}}\nolimits}

\newcommand{\dV}{\; dV}

\def\ip<#1,#2>{\left<#1,#2\right>}

\def\tensor{\otimes}

\newcommand{\Reals}{\mathbb{R}}

\newcommand{\confvel}{U}
\newcommand{\fieldsec}{\mathscr F}
\newcommand{\secA}{\mathcal{A}}

\newcommand{\chargec}{\boldsymbol{\varepsilon}}
\def\dc<#1,#2,#3>{\{#1;\;#2,#3\}}

\newcommand{\hamE}{\mathcal E}
\newcommand{\momJ}{\mathcal J}

\usepackage{etoolbox}

  \newcounter{mnote}
  \let\oldmarginpar\marginpar
  \renewcommand\marginpar[1]{\-\oldmarginpar[\raggedleft\scriptsize #1]%
  {\raggedright\scriptsize #1}}

\title{A Phase Space Approach to the Conformal Construction of Non-Vacuum Initial Data Sets in General Relativity}

\date{\today}

\author{James Isenberg\thanks{isenberg@uoregon.edu} \\ University of Oregon \and David Maxwell\thanks{damaxwell@alaska.edu} \\ University of Alaska Fairbanks}

\begin{document}
\newtheorem{theorem}{Theorem}[section]
\newtheorem{meta-theorem}{Meta-Theorem}[section]
\newtheorem{conjecture}[theorem]{Conjecture}
\newtheorem{problem}[theorem]{Problem}
\newtheorem{proposition}[theorem]{Proposition}
\newtheorem{corollary}[theorem]{Corollary}
\newtheorem{remark}[theorem]{Remark}
\newtheorem{lemma}[theorem]{Lemma}
\theoremstyle{definition}
\newtheorem{definition}[theorem]{Definition}
\numberwithin{equation}{section}

\maketitle

\begin{abstract}
We present a uniform (and unambiguous) procedure for scaling the matter fields in implementing the conformal method to parameterize and construct solutions of Einstein constraint equations with coupled matter sources.  The approach is based on a phase space representation of the spacetime matter fields after a careful $n+1$ decomposition into spatial fields $B$ and conjugate momenta $\Pi_B$, which are specified directly and are conformally invariant quantities. We show that if the Einstein-matter field theory is specified by a Lagrangian which is diffeomorphism invariant and involves no dependence on derivatives of the spacetime metric in the matter portion of the Lagrangian, then fixing $B$ and $\Pi_B$ results in  conformal constraint equations that, for constant-mean curvature initial data, semi-decouple just as they do  for the vacuum Einstein conformal constraint equations. We prove this result by establishing a structural property of the Einstein momentum constraint that is independent of the conformal method: For an Einstein-matter field theory which satisfies the conditions just stated, if $B$ and $\Pi_B$ satisfy the matter Euler-Lagrange equations, then (in suitable form) the right-hand side of the momentum constraint on each spatial slice depends only on $B$ and $\Pi_B$ and is independent of the spacetime metric. We discuss the details of our construction in the special cases of the following models: Einstein-Maxwell-charged scalar field, Einstein-Proca, Einstein-perfect fluid, and Einstein-Maxwell-charged dust. In these examples we find that our technique gives a theoretical basis for scaling rules, such as those for electromagnetism, that have worked pragmatically in the past, but also generates new equations with advantageous features for perfect fluids that allow direct specification of total rest mass and total charge in any spatial region.
\end{abstract}

\section{Introduction}\label{sec:intro}

Initial data for the Cauchy problem in general relativity
consist of
\newcommand\calh{h} 
a Riemannian metric $\calh$ on an $n$-dimensional 
initial time slice $\Sigma$ 
along with a symmetric $(0,2)$-tensor $K$ that
indirectly encodes the momentum conjugate to $\calh$ \cite{arnowitt_dynamical_1959}
and that directly specifies the second 
fundamental form of the embedding of $\Sigma$ into the ambient 
spacetime $M$.
For reasons
analogous to those leading to the Gauss equation of electromagentism, initial data $(\calh,K)$ for the Cauchy problem
cannot be freely specified, but are instead subject to 
the Einstein constraint equations
\begin{align}
R_{\calh} - |K|_{\calh}^2 + (\tr_{\calh} K)^2 &= 16 \pi \hamE &\text{\small [Hamiltonian constraint]}
\label{eq:ham-constraint}\\
\div_{{\calh}} K - d( \tr_{\calh} K) &= -8\pi \momJ
&\text{\small [momentum constraint]}
\label{eq:mom-constraint}
\end{align}
where $R_{\calh}$ and $\tr_{\calh}$ are the scalar curvature and trace
operator of ${\calh}$, $d$ is the exterior derivative on $\Sigma$ 
and, as recalled in Section \ref{sec:constraint-equations}, $\hamE$ 
and $\momJ$ are the energy and
momentum densities of the matter fields.

There is a range of methods for constructing
solutions $({\calh},K)$ of the constraint equations
\eqref{eq:ham-constraint}--\eqref{eq:mom-constraint},
including gluing \cite{chrusciel_mapping_2003}\cite{chrusciel_initial_2005}\cite{corvino_asymptotics_2006}
 and parabolic techniques\cite{bartnik_quasi-spherical_1993}, 
but the conformal method \cite{Lichnerowicz:1944}\cite{york_conformally_1973}
\cite{york_conformal_1999}\cite{pfeiffer_extrinsic_2003}\cite{maxwell_conformal_2014} stands out as a versatile workhorse and is
widely used, for example, in numerical relativity.  It
is especially effective for generating constant-mean curvature
(CMC) solutions, those satisfying $\tr_{\calh} K$ is constant.
Subject to this restriction, the conformal method
provides a satisfactory parameterization
of the set of vacuum solutions of the constraint equations in
vacuum in a number of asymptotic settings 
\cite{Isenberg:1995bi}\cite{york_conformal_1999}\cite{AnderssonChrusciel1996}.  
 
The aim of this paper is to demonstrate how to effectively
apply the conformal method to construct CMC \textit{non-vacuum}
initial data sets. Of
course, the conformal method has long been used in 
the non-vacuum case (\cite{isenberg_extension_1977} is an early reference), but the community
has not standardized on a single approach. Indeed
alternative methods of scaling and unscaling sources appear in the literature \cite{choquet-bruhat_general_2009}.  Moreover, within the method of scaling
sources, there are choices that need to be made for how
various matter fields are scaled.  Perfect fluids
are a notable example, where the approaches 
used in the literature to pick a scaling
(e.g., \cite{dain_initial_2002} or our own previous work \cite{isenberg_gluing_2005}) can 
seem ad-hoc, motivated by analytical practicality rather
than a unifying principle. In this work
we present a uniform and unambiguous approach to applying the 
method of scaling sources to the conformal method 
with coupled matter source fields. As we discuss below, this approach is closely tied to the Hamiltonian formulation of the matter source fields, with the phase space representation of the initial data for the source field playing a key role.  Applying
our technique, we recover well-known scaling
formulas for certain common matter models (electromagentism, most
importantly), but also derive new equations for perfect fluids
and, in particular, even for dust.   That is, our work provides a theoretical
foundation for what has worked practically in the past, and also
derives previously unknown equations that extend the reach of
the conformal method.
Section \ref{secsec:EMCD} shows, for example, how to apply our techniques
to charged dust coupled to electromagnetism and ensure that
in the resulting solution of the constraint equations,
the fluid's charge to mass ratio is constant; previous applications
of the conformal method to fluid matter models 
are unable to ensure this final condition on the charge density.
  We note additionally
that our central results 
(Theorems \ref{thm:form-of-J-manifold-edition} and \ref{thm:form-of-J-tensor-edition})
are not specific to the conformal method; they concern the structure of the momentum density $\momJ$ and are of independent interest.

To illustrate the situation, first consider the 
case of vacuum initial data on a compact
time slice $\Sigma$.  One
can take the gravitational seed data for the conformal method
to consist of
a Riemannian metric $h$, a mean curvature function $\tau$,
a positive function $N$, and a symmetric, trace-free (0,2)-tensor $\confvel$;
we leave this data unmotivated for the moment and defer
a more thorough discussion to Section \ref{sec:constraint-solving}.
One then seeks a positive conformal factor $\phi$ and a vector 
field $Z$ satisfying the Lichnerowicz-Choquet-Bruhat-York (LCBY)
equations
\begin{align}
-2\kappa q\, \Delta_h \phi + R_h \phi - \frac{1}{4N^2}\left|\confvel+\ck_h Z \right|_h^2 \phi^{-q-1}
+ \kappa \tau^2\phi^{q-1} &= 0 &\text{\small [LCBY Hamiltonian constraint]}\label{eq:ham-basic}\\
\div_h\left(\frac{1}{2N}\ck_h Z\right) - \phi^q\kappa d\tau &= -\div_h\left(\frac{1}{2N}\ck_h \confvel\right)
&\text{\small [LCBY momentum constraint]}
 \label{eq:mom-basic}
\end{align}
where $\Delta_h$ is the Laplacian of $h$, $\ck_h$
is the conformal Killing operator (equation \eqref{eq:conformal-killing} below), and where $\kappa$ and $q$
are constants depending on the dimension $n\ge 3$ of $\Sigma$:
\begin{equation}
q = \frac{2n}{n-2},\qquad \kappa = \frac{n-1}{n}.
\end{equation}
For each solution $(\phi,Z)$ of equations \eqref{eq:ham-basic}--
\eqref{eq:mom-basic} there is a corresponding physical solution
$(h^*,K^*)$ of the constraint equations with $h^* = \phi^{q-2}h$
and with $K^*$ determined from $h$, $\phi$, $\confvel$, $N$, $\tau$, and $Z$;
Section \ref{secsec:ConfMethod} contains the details.

When $\tau$ is constant 
system \eqref{eq:ham-basic}--\eqref{eq:mom-basic}
simplifies considerably because the LCBY momentum constraint
no longer depends on $\phi$ and can be first solved for $Z$.
On a compact manifold, there is a unique solution of the
LCBY momentum constraint, up to the addition of 
a conformal Killing field (i.e., an element of the kernel of $\ck_h$).  Moreover, because $Z$ only appears in the LCBY Hamiltonian constraint
via $\ck_h Z$, this potential degeneracy is unimportant, and
the analysis of the existence and multiplicity of solutions 
hinges only on the (scalar) LCBY Hamiltonian constraint. This
decoupling of the momentum constraint 
is a key feature of the conformal method, and
leads to a satisfactory parameterization of the CMC solutions 
of the constraint equations on
compact manifolds where generic seed data leads to a 
unique solution and the exceptional
cases are predictable and well-understood \cite{Isenberg:1995bi}. 
A similar analysis applies in the asymptotically Euclidean
and asymptotically hyperbolic settings as well \cite{york_conformal_1999}\cite{AnderssonChrusciel1996}.
When $\tau$ is not constant, however, the conformal method
appears to suffer from a number of undesirable features \cite{maxwell_model_2011}\cite{nguyen_nonexistence_2018}\cite{1710.03201},
and for this reason we focus our attention on the CMC case.
See also \cite{maxwell_initial_2021} for an alternative generalization of
the CMC conformal method to non-CMC initial data.

In the presence of matter, the LCBY equations become
\begin{align}
-2\kappa q \Delta_h \phi + R_h \phi - \frac{1}{4N^2}\left|\confvel+\ck_h Z \right|_h^2 \phi^{-q-1}
+ \kappa \tau^2\phi^{q-1} &= 16\pi \phi^{q-1} \hamE^*\label{eq:ham-basic-matter}\\
\div_h\left(\frac{1}{2N}\ck_h Z\right) - \phi^q\kappa d\tau &= -\div_h\left(\frac{1}{2N}\confvel\right)
-8\pi \phi^{q} \momJ^*\label{eq:mom-basic-matter}
\end{align}
where we have decorated the physical
energy and momentum densities $\hamE^*$ and $\momJ^*$
with stars to indicate that they 
are generally expressions that depend on the solution metric $h^*$ (and hence on $h$ and $\phi$) along with certain
seed data related to the matter model.  For example,
in the case of a scalar field we can take the matter seed data
to consist of the field value $\psi$ along with a function
$P$ that prescribes the normal derivative of the scalar field in 
the ambient spacetime.  In this case,
\begin{align}
\hamE^* &=  \frac{1}{2}P^2 + \frac{1}{2}| d \psi|^2_{h^*}\\
& = \frac{1}{2}P^2 + \frac{1}{2}\phi^{2-q}| d \psi|^2_h,\nonumber \\[8pt]\momJ^* &=  -P\;  d\psi.\label{eq:j-star-SF-try1}
\end{align}
However, the LCBY momentum constraint now no longer decouples from
the rest of the system because of the term
$\phi^{q} \momJ^* =  -\phi^q P  d\psi$ appearing in
equation \eqref{eq:mom-basic-matter}.
Pragmatically, the obvious approach to 
recover the decoupling is to allow $P$ to conformally transform
so that the scaling $h^* = \phi^{q-2}h$ is associated
with a scaling $P^* = \phi^{-q} P$ for some initially 
specified function $P$.  Replacing $P$ with $P^*$ 
in equation \eqref{eq:j-star-SF-try1}, and assuming additionally
that $\tau$ is constant, equations \eqref{eq:ham-basic-matter}--\eqref{eq:mom-basic-matter} become
\begin{align}
-2\kappa q \Delta_h \phi + R_h \phi - \frac{1}{4N^2}\left|\confvel+\ck_h Z \right|_h^2 \phi^{-q-1}
+ \kappa \tau^2\phi^{q-1} &= 8\pi \left[ \phi^{-q-1} P^2
+ \phi | d\psi|_h^2\right]\label{eq:scalar-field-Ham}\\
\div_h\left(\frac{1}{2N}\ck_h Z\right) &= -\div_h\left(\frac{1}{2N} \confvel\right) +8\pi P  d\psi\label{eq:scalar-field-mom}
\end{align}
which achieves the desired decoupling of the momentum constraint.
On a compact manifold, in the generic case where $h$ has no conformal Killing fields, equation \eqref{eq:scalar-field-mom} 
can 
be solved independently for $Z$, and there is a unique solution.  
Thus, as in the vacuum setting, the analysis reduces to the
Lichnerowicz equation.  It is worth noting that
although the decoupling simplifies the analysis, the resulting
Lichnerowicz equation can still harbor interesting features.  
Indeed, for the scalar field LCBY Hamiltonian constraint \eqref{eq:scalar-field-Ham}
the term
$\phi | d\psi|^2_h$ 
interacts nontrivially
with the scalar curvature term and we refer to, e.g, 
\cite{hebey_variational_2008}
 for an analyses.

A similar scaling/decoupling procedure can be applied to
electromagnetism. Now
the matter fields can be specified by a 1-form $E$ (the electric
field) and a 2-form $\mathcal B$ (the magnetic field).  In this case
\begin{align*}
\hamE^* &= \frac{1}{8\pi}\left[ |E|_{h^*}^2 + \frac{1}{2}|\mathcal B|_{h^*}^2\right]\\
&= \frac{1}{8\pi}\left[ \phi^{2-q}|E|_h^2 + \frac{1}{2}\phi^{4-2q}|\mathcal B|_h^2\right],\\[8pt]
\momJ^* &= 
\frac{1}{4\pi} \ip<E,\cdot\interior \mathcal B>_{h^*}\\
&= \frac{1}{4\pi} \phi^{2-q}  \ip< E,\cdot\interior \mathcal B>_h.
\end{align*}
The right-hand side of the LCBY momentum constraint then contains a term
$\phi^2 \ip< E,\cdot\interior \mathcal B>$
which again leads to coupling of the full LCBY system, even in 
in the CMC setting.  To circumvent this difficulty
we introduce a conformally scaling
electric field so that the transformation 
$h^* = \phi^{q-2}h$ is associated
with a transformation $E^* = \phi^{-2} E$ for some
initially specified 1-form $E$.  Moreover,
$\div_{h^*} E^* = \div_{h^*} \phi^{-2} E = \phi^{-q} \div_h E$
and hence if $E$ is $h$-divergence free (as is required by
Gauss' Law when there are no electric sources), then
$E^*$ is $h^*$-divergence free as well.  It is seemingly remarkable
that the scaling law that leads to decoupling the LCBY equations
simultaneously leads to preservation of the electrovac constraint
$\div_{h^*} E^* = 0$.

The success of the scaling/decoupling procedure for these and
other matter models naturally leads to the question of whether
there is an underlying principle that unifies these examples.
Moreover, one would like to know if there is a rule that
determines directly (without ad-hoc manipulation)
how to scale the matter fields in general. Our main results give 
a positive resolution to these questions under 
the following physically natural hypotheses:
\begin{itemize}
  \item the matter fields have an associated Lagrangian,
  \item the matter Lagrangian is diffeomorphism invariant,
  \item the matter Lagrangian is minimally coupled to gravity
  (so no derivatives of the spacetime metric appear in it).
\end{itemize}
A Lagrangian in this setting depends on the spacetime metric $g$ along
with matter fields $\fieldsec$ and, as recalled in Section \ref{sec:constraint-equations}, the momentum
density $\momJ$ is a well-defined quantity 
on a time slice $\Sigma$ depending on 
$g$ and $\fieldsec$ regardless of whether the Einstein equations
or the equations of motion for the matter fields are satisfied.
On this same slice, the spacetime metric determines a Riemannian
metric $h$, and our main results stem from an invariance
property of the combination $\momJ\dV_h$ under metric perturbations.
Indeed, in Sections \ref{sec:manifold-valued} and \ref{sec:tensor-valued}
we prove two theorems that are concrete 
expressions of the following imprecise meta-theorem, which is 
a structural statement about the Einstein equations and is
not simply a fact about the conformal method.
\begin{meta-theorem}\label{mt:meta-theorem}
Under the above hypotheses, the spacetime matter fields can be decomposed
into certain `spatial' variables $B$ 
with conjugate momenta\footnote{The conjugate momentum of a field
is the derivative of a Lagrangian with respect to the field's time 
derivative; formal definitions appear just prior 
to Theorem \ref{thm:form-of-J-manifold-edition} 
and Corollary \ref{cor:vary-X-to-Lie} below.}
$\Pi_B$ such that
if $B$ and $\Pi_B$ satisfy the Euler-Lagrange
equations of the matter Lagrangian, then
the covector-valued $n$-form
\begin{equation}\label{eq:J-density}
\momJ dV_h
\end{equation}
on each time slice is
determined by $B$ and $\Pi_B$ alone, and is independent of the metric. 
\end{meta-theorem}

That is, although the spatial metric volume form $dV_h$ and the momentum
density $\momJ$ each depend on the the metric, 
the combination $\momJ dV_h:= \momJ\otimes dV_h$ 
depends only on the state and momenta of the matter fields
after a suitable decomposition into `spatial' variables.
In fact, we give an explicit relation for how 
$\momJ dV_h$ is determined from the matter fields,
equation \eqref{eq:J-to-Pi-on-B} below, but these details
are not needed to see how 
the meta-theorem resolves the question of decoupling.
Indeed, consider a pair $h$ and $h^*$ of 
conformally related metrics: $h^* = \phi^{q-2}h$.
Fixing $B$ and $\Pi_B$, the meta-theorem implies
\[
\momJ^* dV_{h^*} = \momJ dV_{h},
\]  
where $\momJ^*$ and $\momJ$ are the momentum densities determined by the fixed matter fields along with 
$h^*$ and $h$ respectively.  
A computation shows $dV_{h^*}=\phi^{q}dV_h$ and 
as a consequence
$\momJ^* = \phi^{-q} \momJ$.
Thus, so long as $B$ and $\Pi_B$ remain fixed,
the term $\phi^{q}\momJ^*$ appearing
on the right-hand side of the LCBY momentum constraint
\eqref{eq:mom-basic-matter}
is the unique value $\momJ=\momJ(h,B,\Pi_B)$ 
regardless of the choice 
of conformal factor, and this leads to decoupling of the momentum 
constraint for CMC data.

The scaling law for the matter terms $B$
and $\Pi_B$ is the simplest one imaginable:
scale nothing.  Matter seed data consists of fixed field
values and conjugate momenta, and in this sense
the method of scaling sources is something of a misnomer:
presented properly, the method would be better described as
that of 
\textbf{\textit{non-}}scaling sources, except that this name
is already used in the literature for something quite different 
and less practical.  One may wish to 
represent the matter field in terms
of different variables that mix in the metric, in which
case there would be a non-trivial conformal scaling.
In the case of a scalar field, it turns out that 
the momentum conjugate
to $\psi$ is $\Pi_\psi = 2P\, dV_h$ and hence if $\Pi_\psi$
remains fixed but $h$ conformally transforms to $h^*=\phi^{q-2} h$,
then 
$P$ must conformally transform to $P^*=\phi^{q} P$, exactly
the scaling law used to derive equations \eqref{eq:scalar-field-Ham}--
\eqref{eq:scalar-field-mom}.  For electromagnetism, the important
dynamical field field is a $U(1)$ connection represented locally 
on a time slice by a spatial 1-form $A$, and its curvature is the
2-form $\mathcal B$.  The momentum conjugate to $A$ is
the vector valued $n$-form
$\Pi_A = -1/(4\pi)\left<E,\cdot\right>_h dV_h$.
One readily verifies that if $\Pi_A$ remains fixed and
$h$ conformally transforms to $\phi^{q-2}h$, then we
must conformally transform $E$ to $E^*=\phi^{-2} E$, which
is again the scaling law used in the discussion prior
to Meta-Theorem \ref{mt:meta-theorem}.  Gauss' Law $\div_h E=0$
is, in fact, a constraint on $\Pi_A$, and since $\Pi_A$ remains
fixed, Gauss' Law remains satisfied.

The vagueness of Meta-Theorem \ref{mt:meta-theorem}
comes from the breadth of 
what can be included under the umbrella
of physical field theory.  We prove 
variations of the meta-theorem in the following cases:
\begin{enumerate} 
\item The matter
fields take on values in a separate
manifold (Theorem \ref{thm:form-of-J-manifold-edition}).
\item The matter fields are tensor-valued (Theorem \ref{thm:form-of-J-tensor-edition}).
\end{enumerate}
Certain matter models, such as 
electromagnetism coupled with a charged fluid, involve
both tensor-valued and manifold-valued fields and 
can be addressed
by a straightforward generalization of our results.  
Nevertheless, the framework for field theory that
we use is not enough to include all conceivable cases and
the Vlasov matter model is a notable exception that we will address
in future work.

The mechanism of how $\momJ dV_h$ is determined 
from $B$ and $\Pi_B$ in Meta-Theorem \ref{mt:meta-theorem}
and its precise forms Theorems \ref{thm:form-of-J-manifold-edition}
and \ref{thm:form-of-J-tensor-edition}
is a little subtle.
In Proposition \ref{prop:X-by-Lie-only} concerning tensor-valued
fields, we show that so long as $B$ and $\Pi_B$
solve the Euler-Lagrange equations on the ambient manifold, then 
for any compactly supported vector field $\delta X$ on 
a time slice $\Sigma$
\begin{equation}\label{eq:J-to-Pi-on-B}
\int_{\Sigma} \momJ(\delta X) dV_h = -\frac{1}{2}
\int_{\Sigma} \Pi_B( \Lie_{\delta X} B).
\end{equation}
One would
like to claim that
\begin{equation}\label{eq:almost-J-to-Pi-B}
\momJ dV_h = -\frac{1}{2} \Pi_B( \Lie B)
\end{equation}
but $\Lie_{\delta X}$ is not tensorial in $\delta X$
and an exact expression for $\momJ dV_h$ usually
requires some integration by parts.  The hypothesis that
$B$ and $\Pi_B$ solve the Euler-Lagrange equations on the ambient
manifold is seemingly out of place for a theorem about initial
data, where $B$ and $\Pi_B$ are specified on a single time slice.
However, the Euler-Lagrange equations may imply constraints on 
the initial data, such as Gauss' law, and it turns out that 
these constraints must be satisfied 
in order for equation \eqref{eq:J-to-Pi-on-B} to hold;
see Section \ref{sec:EMCSF} where this is illustrated
explicitly for a charged scalar field coupled to electromagnetism.
On the other hand, the situation for manifold-valued fields 
is somewhat simpler; no $(n+1)$-decomposition of the spacetime
matter field $\fieldsec$ is involved and in Proposition \ref{prop:X-by-push-only} we show
\[
\momJ dV_h = -\frac{1}{2} \Pi_{\fieldsec}\circ \fieldsec_*.
\]

Some historical remarks concerning our results are in order.
In fact, the general notion that matter fields and
conjugate momenta should remain fixed in order to obtain decoupling
was known in the 1970s, and an imprecise statement along
the lines of equation \eqref{eq:almost-J-to-Pi-B} appears
in earlier work by one of us \cite{isenberg_extension_1977}.  Although
\cite{isenberg_extension_1977} cites the physics-style paper 
\cite{kuchar_kinematics_1976} as
justification, the literature does not appear 
to contain a mathematical proof.  Indeed, there are important considerations needed to arrive at equation
\eqref{eq:almost-J-to-Pi-B}, including a specific form
of the $(n+1)$-decomposition, that appear to have been 
previously overlooked.
Moreover, the underlying principle 
seems to have been forgotten by the community generally, and is not
mentioned in modern texts, e.g. \cite{choquet-bruhat_general_2009}.  Indeed, it was
recently independently rediscovered by one of us (Maxwell).  
The current work
therefore serves two purposes: first, to advertise the principle
and second to give it rigorous justification.

The remainder of the paper has the following structure:
\begin{itemize}
  \item Section \ref{sec:prelim} contains preliminaries
  to fix notation and recall important facts about the constraint equations.  This can be skimmed, taking into account our 
  notation $j_*$ 
  for a projection defined in Section \ref{sec:slicing}
  that is non-standard but used pervasively in the subsequent sections.
  \item Section \ref{sec:manifold-valued} contains
  Theorem \ref{thm:form-of-J-manifold-edition},
  the version of Meta-Theorem \ref{mt:meta-theorem}
  that applies to matter sources taking on values in
  a separate manifold. While this is not a common setting,
  fluids are a notable application. 
  \item Section \ref{sec:tensor-valued} contains
  Theorem \ref{thm:form-of-J-tensor-edition},
  the version of Meta-Theorem \ref{mt:meta-theorem}
  that applies to tensor-valued matter sources.  
  \item Section \ref{sec:EMCSF} illustrates our main results
  in the context of electromagnetism coupled to a charged
  scalar field (EMCSF).  It serves as a concrete example
  to motivate the somewhat technical details of 
  Section \ref{sec:tensor-valued}, especially the assumptions
  needed to obtain equation \eqref{eq:J-to-Pi-on-B}.
  \item Section \ref{sec:constraint-solving} provides
  a more in-depth overview of the conformal method,
  and links it to a technique for solving a constraint
  appearing in the EMCSF model of Section \ref{sec:EMCSF}.
  Section \ref{secsec:CMC} contains remarks about
  important differences, even with the good scaling
  provided by Meta-Theorem \ref{mt:meta-theorem},
  between the non-vacuum CMC 
  conformal method versus its vacuum formulation.
  \item Finally, Section \ref{sec:apps} contains applications
  of our results in three interesting cases in the context
  of the conformal method.  In Section
  \ref{sec:fluids} we treat perfect fluids and, using
  the Lagrangian formulation for fluids, derive
  equations different from those found in \cite{dain_initial_2002}
  or \cite{isenberg_gluing_2005},
  including a new equation for dust. Our construction 
  directly prescribes the total particle number in every spatial
  region, whereas past constructions do not preserve even
  the total particle number.
  Section \ref{secsec:proca}
  applies our results to the Proca model of a massive, spin-1
  particle.  The resulting equations contain an
  unexpected and potentially insightful appearance 
  of the `densitized lapse', perhaps the least-well understood
  component of the conformal method. Finally, Section
  \ref{secsec:EMCD} considers electromagnetism coupled to
  charged dust.  Because we directly specify both the total rest mass and charge in any spatial region, we ensure that the mass to charge ratio of the particles
  is constant, a distinguishing feature of our approach.
\end{itemize}

\section{Preliminaries}
\label{sec:prelim}
\subsection{Slicing and Spatial Tensors} 
\label{sec:slicing}
The focus of this work is essentially local and we
therefore take, once and for all, 
$M=\Sigma\times I$ where
$\Sigma$ is an open subset of $\Reals^n$ with
coordinates $x^i$ and where $I$ is an open interval
in $\Reals$ with coordinate $t$.  Each $t\in I$
is associated with the time slice $\Sigma_t = \Sigma\times\{t\}$.
The split into space and time, while involving an arbitrary choice, is
essential to the Hamiltonian approach that underlies our approach.
With these fixed coordinates in hand, we use them whenever expedient
to shorten a definition or argument, rather than relying on coordinate-free
notations.

Although our main concern is initial data, our construction is
tightly linked to the evolution problem and for the most
part we work with tensors on the ambient spacetime $M$.
Tensor fields on $M$ are assumed, for simplicity, to be smooth.
A tensor $T$ on $M$ is \textbf{spatial} if
it vanishes when one of its arguments is 
either of $dt$ or $\partial_t$.
Equivalently, a spatial tensor consists of tensor products
of $dx^i$ and $\partial_{x_i}$, and spatial tensors define,
unambiguously, tensors intrinsic to each slice $\Sigma_t$.

We define a projection $j_*$
onto spatial tensors by declaring
\begin{align*}
j_*\; dt &=0,\\
j_*\; dx^i &= dx^i,\\
j_*\; \partial_t &= 0,\\
j_*\; \partial_{x^i} &= \partial_{x^i}
\end{align*}
and by
extending $j_*$ to higher rank tensors by $j_* T(\cdot,\cdots,\cdot$)
= $T(j_*\cdot,\cdots,j_*\cdot)$. In coordinates, $j_* T$ is obtained
from $T$ by dropping any component containing either $dt$ or 
$\partial_t$.  

The following elementary observations
can all be established by working in coordinates.
\begin{lemma}\label{lem:basic-split-Lie}\strut
\begin{enumerate}
  \item\label{part:lie-of-spatial}
  Suppose $B$ is a spatial tensor field.  Then 
  $\partial_t B:=\Lie_{\partial_t} B$ is spatial, as is $\Lie_S B$ for any  spatial vector field $S$.

  \item\label{part:lie-of-restrict}
  For any tensor field $T$,
  $\partial_t j_* T = j_* \partial_t T$.

  \item\label{part:lie-of-zero}
   If $T$ is a tensor field that vanishes
  on a slice $\Sigma_{t}$, then for any spatial vector
  field $S$, $\Lie_S T = 0$ on $\Sigma_{t}$. 

  Consequently, if $T_1$ and $T_2$ are arbitrary
  tensor fields that are equal on a slice $\Sigma_{t}$, then
  $
  \Lie_S  T_1 = \Lie_S T_2
  $
  on $\Sigma_{t}$.
\end{enumerate}
\end{lemma}

Note that the projection $j_*$ is a coordinate-dependent operation,
but is independent of any metric.  

The Hamiltonian approach to the evolution problem in general relativity
requires splitting tensor fields into spatial components defined 
on the slices $\Sigma_t$, and this applies equally to the
spacetime metric. A Lorentzian metric $g$ on $M$ is 
\textbf{slice compatible}
if the restriction of $g$ to each slice $\Sigma_t$
is Riemannian. A slice-compatible Lorentzian metric 
has a timelike unique normal vector field $\nu$ defined by
$\nu(dt)>0$, $g(\nu,\nu)=-1$ and 
$j_*\; g(\nu,\cdot) =0$.
The vector field $\partial_t$ can then be written
\[
\partial_t = N \nu + X
\]
where the \textbf{shift} $X$ is spatial and where the \textbf{lapse}
$N$ satisfies $\nu(dt) = 1/N$. The $n+1$ \textbf{decomposition} of
$g$ consists of the spatial tensors $h = j_*\; g$,
$N$ and $X$.  As is clear from the coordinate representation
$g = -N^2dt^2 + h_{ab}(dx^a + X^adt)(dx^b+X^bdt)$, the metric
is uniquely determined by its $n+1$ decomposition, and
we write $g(h,N,X)$ for the Lorentzian metric reconstructed
from a spatial Riemannian metric $h$, a lapse $N>0$
and a spatial vector $X$. 
Our notation here conflicts mildly with that used
in the introduction, where $h$ and $dV_h$ are tensor fields
defined on a single time slice $\Sigma_t$, whereas here they are spatial
tensors defined on all of $M$.  Moreover, $h$ is not a Riemannian metric in the usual sense because $h(\partial_t,\cdot)=0$, but it becomes
one when restricted to any slice $\Sigma_t$.

The manifold $M$ is orientable and we fix an orientation
on it. Let
$dV_g$ be the positively oriented volume form of $g$ on $M$.
We then define
\[
dV_{h} = N^{-1} \partial_t\interior dV_g
\]
and observe that $dV_h$ is spatial and $N\, dt\wedge dV_h = dV_g$.

\subsection{Einstein Constraint Equations}\label{sec:constraint-equations}
In this section we give a brief overview of the origin of 
the Einstein constraint equations, with a particular focus
on deriving relation \eqref{eq:def-J-first} below, which is
a key tool in our work.

A Lagrangian for a field theory on $M$ is defined in terms of
an $n+1$ form on $M$.  In particular, 
the Einstein-Hilbert Lagrangian for a Lorentzian metric $g$ is
\[
\mathcal L_{EH}(g,\partial g, \partial^2 g) = R_g dV_g
\]
where $R_g$ is the scalar curvature of $g$.
Recall that we have a fixed frame $\partial_k = \partial_{x^k}$ 
of vector fields, and  $\partial g$
denotes the collection of Lie derivatives $(\partial_1 g,\ldots,\partial_n g,\partial_t g)$, with $\partial^2 g$ defined similarly.

Let $\mathcal L_{\mathrm{Matter}}(\fieldsec,g)$ be a Lagrangian
on $M$ depending on a section $\fieldsec$ 
of some fiber bundle
over $M$ along with a Lorentzian metric $g$. In general,
$\mathcal L$ depends on `derivatives' of $\fieldsec$,
and we are more precise about this in applications.  Regarding
the metric, however, we assume the matter Lagrangian
is minimally coupled and therefore depends pointwise
on the values of $g$ but not its derivatives.
The total Lagrangian for the metric and matter is then
\[
\mathcal L_{EH} + 8\pi \mathcal L_{\mathrm{Matter}}.
\]

Recall that the action for the total Lagrangian is stationary at $g$
if it satisfies the Einstein equation
\[
G = 8\pi T
\]
where $G_g = \Ric_g - R_g g$ is the Einstein tensor
and where (as in \cite{wald:1984} Appendix E) the stress-energy
tensor $T$ is defined by varying $\mathcal L_{\mathrm{Matter}}$
with respect to the inverse metric $g^{-1}$:
\[
T dV_g = -\frac{\partial \mathcal L_{\mathrm{Matter}}}{\partial g^{-1}}
\]
More explicitly, we are considering the map
\[
g^{-1} \mapsto \mathcal L_{\mathrm{Matter}}(\fieldsec, g(g^{-1}))
\]
pointwise on $M$ and linearizing with respect to $g^{-1}$.
Hence $\partial \mathcal L_{\mathrm{Matter}}/\partial g^{-1}$
can be identified with a symmetric (0,2) tensor-valued $n+1$-form,
and $T$ is obtained by dividing by $-dV_g$.

An inverse metric $g^{-1}$ that is Riemannian on the slices
$\Sigma_t$ can be written as a function of
a spatial Riemannian inverse metric $h^{-1}$ and a vector field 
$\nu$ transverse to the slices:
\[
g^{-1} =  - \nu\otimes \nu + h^{-1}.
\]
With respect to these variables,
\[
\frac{\partial g^{-1}}{\partial \nu}[\delta \nu] = 
-\left(\nu \otimes \delta \nu + \delta \nu \otimes \nu\right) 
\]
and the Einstein constraint equations arise from variations
of the total action with respect to $\nu$:
\begin{equation}\label{eq:ECE-abstract}
G(\nu,\cdot) = 8\pi\, T(\nu,\cdot).
\end{equation}

The scalar \textbf{energy density} $\hamE$ and covector
\textbf{momentum density} $\momJ$ 
of the matter fields seen by an observer with tangent $\nu$ are, by definition,
\begin{align}
\hamE &= T(\nu,\nu)\label{eq:E-def-direct}\\
\momJ &= -j_*T(\nu,\cdot)\label{eq:J-def-direct}.
\end{align}
  
On a slice $\Sigma_t$ we have the induced
metric $h=j_* g$ and the second fundamental
form $K=\frac12 \Lie_\nu h$.  The Gauss 
and Codazzi equations allow the the left-hand side
of equation \eqref{eq:ECE-abstract} to be written
in terms of $g$ and $K$, whereas the right-hand
side is written in terms $\hamE$ and
$\momJ$ essentially by fiat. 
The Einstein constraint equations 
\eqref{eq:ham-constraint}--\eqref{eq:mom-constraint} then arise from identifying
$\Sigma$ with $\Sigma_t = (\iota_t)_*\Sigma$
and spatial tensors on $\Sigma_t$ with their pullbacks onto
$\Sigma$.

Writing $\nu = (\partial_t - X)/N$,
the momentum density $\momJ$ can be computed
directly by varying $\mathcal L_{\mathrm{Matter}}$
with respect to the shift $X$ and leaving $h$ (or equivalently $h^{-1}$) 
and $N$ fixed. Indeed, unwinding definitions we 
compute the fundamental identity
\begin{equation}\label{eq:def-J-first}
\begin{aligned}
\frac{\partial\mathcal{L}_{\mathrm{Matter}}}{\partial X}[\delta X] &= 
-2 T(\nu,\delta X/N) \, dV_g\\
&= -2 T(\nu,\delta X)\,  dt\wedge dV_h \\
&= 2\,\momJ(\delta X)\, dt\wedge dV_h
\end{aligned}
\end{equation}
for any spatial vector field $\delta X$.  A related computation
shows
\begin{equation}\label{eq:def-E-first}
\frac{\partial\mathcal{L}_{\mathrm{Matter}}}{\partial N}
= -2 \hamE \, dt\wedge dV_h.
\end{equation}

\section{Lagrangians depending on manifold-valued fields}
\label{sec:manifold-valued}

In this section we prove Theorem \ref{thm:form-of-J-manifold-edition},
the concrete version of Meta-Theorem \ref{mt:meta-theorem}
in the case where matter sources take on values in a 
separate manifold $F$ unrelated to spacetime itself.  

Loosely, a Lagrangian in this context consumes a Lorenzian 
metric $g$ together with a map $\mathcal F$ from $M$ to $F$ 
and yields
an $n+1$ form $\mathcal L(\mathcal F, \partial \mathcal F, g)$
that depends on the values and first derivatives of $\mathcal F$.
To formalize this construction,  let $E$ be the bundle
over $F$ consisting of $n+1$ Whitney sums
\[
E = \underbrace{TF\oplus\cdots\oplus TF}_{\text{$n+1$ times}}.
\]
A \textbf{minimially coupled Lagrangian with field values from} $F$ 
is a bundle map
\[
\mathcal L : (E\times M) \times_M G(M) \to \Lambda^{n+1}(M)
\]
where 
$E\times M$ is the trivial bundle over $M$
with fiber $E$,
$G(M)$ is the bundle of Lorentzian metrics,
$\Lambda^{n+1}(M)$ is the bundle of $n+1$-forms
and $\times_M$ denotes the fiber product.
In practice, we take the domain of $\mathcal L$
to be the subbundle of $(F\times M)\times_M(E\times M) \times_M G(M)$
where the first factor records the common base point in $F$
of the tangent vectors from the second factor.

A section $\fieldsec$ of $F\times M$ (i.e., a map from $M$ to $F$) and 
a Lorentzian metric $g$ determine a field of $n+1$ forms
\[
\mathcal L(\fieldsec, \partial\fieldsec, g)
\]
where $\partial\fieldsec$ is shorthand for $(\fieldsec_* \partial_1,\ldots,\fieldsec_*\partial_n,\fieldsec_*\partial_t)$.  The Lagrangian is \textbf{diffeomorphism invariant} if whenever $\Psi$ is an
orientation preserving diffeomorphism from some neighborhood 
$U_1$ of  $M$ to a neighborhood $U_2$ of $M$, then for all
$\fieldsec$ and $g$ defined on $U_2$,
\[
\Psi^*\mathcal L(\fieldsec, \partial\fieldsec, g)
= \mathcal L(\Psi^*\fieldsec, \partial \Psi^*\fieldsec, \Psi^* g)\quad\text{on $U_1$}.
\]

\newcommand{\dotfieldsec}{
\mathrlap{\skew{8}\dot{\raisebox{1pt}{\phantom{$\mathscr F$}}}}{\mathscr F}}

The spacetime Lagrangian reduces to a \textbf{slice Lagrangian}
on slices $\Sigma_t$ as follows.
Let $p\in M$, and let 
\begin{enumerate}
\setlength{\itemsep}{0pt}%
  \item $h$ be a spatial Riemannian metric at $p$,
  \item $N>0$,
  \item $X$ be a spatial vector at $p$
  \item $\fieldsec\in F$,
  \item $\fieldsec_{(i)}\in T_\fieldsec F$, $1\le i\le n$,
  \item $\dotfieldsec\in T_\fieldsec F$.
\end{enumerate}
The associated value of the \textbf{slice Lagrangian} at $p$ is then
\[
L_\Sigma(\fieldsec,(\fieldsec_{(1)},\ldots,\fieldsec_{(n)}),\dotfieldsec, h,N,X) = \partial_t\interior \mathcal L(\fieldsec,(\fieldsec_{(1)},\ldots,\fieldsec_{(n)},\dotfieldsec),g(h,N,X));
\]
the primary distinction between the bundle maps 
$\mathcal L$ and $L_\Sigma$
is that all tensors associated with $L_\Sigma$, including 
its value as an $n$-form, are spatial. Additionally, the arguments
corresponding to spatial derivatives of $\fieldsec$
are separated from the time derivative, and we
use the notation $\partial_x \fieldsec = (\fieldsec_*\partial_1,\ldots\fieldsec_*\partial_n)$ whenever $\fieldsec:M\to F$.

Proposition \ref{prop:X-by-push-only} below concerns the 
structure of a diffeomorphism-invariant
Lagrangian and its proof employs diffeomorphisms that preserve $t$ (i.e., 
diffeomorphisms $\Psi$ for which
$\Psi^* t = t$). Before proving it, we require the following two technical lemmas concerning such
diffeomorphisms, the first of which describes their interaction
with the projection $j_*$.  Note that here and elsewhere 
we allow $\Psi^*$ to act on vectors by $\Psi^* V=\Psi^{-1}_*V$,
and by extension $\Psi^*$ is well defined on general tensors.

\begin{lemma}\label{lem:t-preserving-diffeo}
Suppose $\Psi$ is a $t$-preserving diffeomorphism from a neighborhood
$U$ of $M$ to a neighborhood $V$ of $M$
\begin{enumerate}
  \item\label{part:covariant-almost-commutes} If $\omega$ is a covariant tensor, then
  \[
  j_* \Psi^* \omega = j_* \Psi^* j_*\omega.
  \]
\item \label{part:kills-S-interior-omega}
If $\omega$ is an $n+1$ form and $S$ is a spatial vector, then
\[
j_* \Psi^* (S\interior \omega) = 0.
\]

  \item\label{part:spatial-stays-put} If 
additionally $\Psi$ is the identity on some slice
$\Sigma_{t_0}\cap U$, and if 
  $B$ is a spatial tensor on $V$, then
  \begin{equation}\label{eq:pull-back-spatial}
j_* \Psi^* B = B\quad \text{on $\Sigma_{t_0}\cap U$.}
  \end{equation}
\end{enumerate}
\end{lemma}
\begin{proof}
We prove part \ref{part:covariant-almost-commutes}
for a 1-form $\omega$; the general case follows
since $j_*$ and $\Psi^*$ distribute over tensor products.
If $S$ is a spatial vector, so is $\Psi_* S$ and
\[
(j_* \Psi^* \omega)(S) = \Psi^*\omega (S) = \omega(\Psi_* S) = (j_* \omega)(\Psi_* S) = \Psi^*(j_* \omega)(S) = (j_* \Psi^* j_* \omega)(S).
\]
We conclude that $j_* \Psi^* \omega = j_*\Psi^* j_*\omega$ since
both vanish on $\partial_t$ as well.

To establish part \ref{part:kills-S-interior-omega},
observe that if $S$ is spatial and $\omega$ is an
$n+1$ form, then $S\interior \omega$ vanishes
when all of its arguments are spatial.  So
$S\interior \omega = dt\wedge \eta$ for some spatial
$n-1$ form $\eta$ and from part \ref{part:covariant-almost-commutes}
we have $j_* \Psi^* (S\interior\omega)
= j_*\, \Psi^* ( j_* dt\wedge \eta) = 0$.

Finally, we turn to part \ref{part:spatial-stays-put}
and assume that $\Psi$ is the identity on $\Sigma_{t_0}\cap U$.
If $S$ is a spatial vector on $\Sigma_{t_0}$ then
$\Psi_* S = S$.  Hence for any 1-form $\omega$,
\[
j_* \Psi^*\omega (S) = 
\Psi^*\omega(S) = \omega(\Psi_* S) = \omega(S).
\]
If in addition $\omega$ is spatial, then $\omega(\partial_t)=0$,
as is $j_* \Psi^* \omega (\partial_t)$, and we conclude that
$j_* \Psi^*\omega = \omega$.
On the other hand, for a spatial vector $V$ we have the obvious identity
$\Psi^* V = V$
on $\Sigma_{t_0}$. Hence $j_* \Psi^* V = j_* V = V$.
This establishes equation \eqref{eq:pull-back-spatial}
if $B$ is a spatial 1-form or vector, and the general result
follows from our earlier observation that $j_* \Psi^*$
distributes over tensor products.
\end{proof}

The following
lemma concerns the $n+1$ decomposition of the pullback of
a metric by a $t$-preserving diffeomorphism, and we note
the appearance of $j_*$ in it because the pullback of 
a spatial 1-form by a $t$-preserving diffeomorphism
need not be spatial.

\begin{lemma}\label{lem:pullback-g}
Suppose $\Psi$ is a $t$-preserving diffeomorphism from
a neighborhood $U$ of $M$ to a neighborhood $V$ of $M$.
If $g$ is a slice-compatible Lorentzian metric on $V$
with $(n+1)$ decomposition $(h,N,X)$ then
$\hat g = \Psi^* g$ is a slice-compatible Lorentzian
metric on $U$ with $(n+1)$ decomposition $(\hat h, \hat N, \hat X)$
satisfying
\begin{align}
\hat h &= j_* \Psi^* h\label{eq:pull-h}\\
\hat N &= \Psi^* N\label{eq:pull-N}\\
\hat X &= \Psi^*(X+S)\label{eq:pull-X}
\end{align}
where $S$ is the spatial vector field satisfying
$
\Psi_* \partial_t = \partial_t + S.
$

We have the additional identities
\begin{equation}\label{eq:push-partial-0}
\Psi^*(\partial_t-X) = \partial_t-\hat X.
\end{equation}
and
\begin{equation}\label{eq:pull-dV}
 dV_{\hat h} = j_* \Psi^* dV_{h}.
\end{equation}
\end{lemma}
\begin{proof}
First observe that if $W$ is a nonzero spatial vector on $U$, then
$\Psi_* W$ is nonzero and spatial on $V$ and hence 
$\Psi^* g$ is a Lorentzian metric satisfying
$\Psi^* g(W,W)= g(\Psi_* W,\Psi_* W) >0$.  So it is also slice compatible.
Moreover, if $W$ and $Z$ are spatial vectors at some point we find
\[
\hat h(W,Z) = \hat g(W,Z) = g(\Psi_*W,\Psi_*Z) = h(\Psi_*W,\Psi_*Z) = (\Psi^* h)(W,Z).
\]
The equality 
$\hat h = j_* \Psi^* h$
now follows since both tensors vanish when some argument is $\partial_t$.

Since $\Psi$ preserves $t$, $\Psi_* \partial_t = \partial_t + S$ for
some spatial vector $S$.  Consequently $\Psi^*\partial_t = \partial_t-\Psi^* S$.  If $\nu$ and $\hat \nu$ are the unit normals associated
with $g$ and $\hat g$ it is easy to see that $\hat \nu = \Psi^* \nu$.
Hence
\[
\hat \nu = \Psi^*\left(\frac{\partial_t - X}{N}\right)
= \frac{\partial_t - \Psi^*(S)-\Psi^*(X)}{\Psi^* N}
\]
and we identify $\hat N = \Psi^*N$ and $\hat X =\Psi^*(X + S)$.

Equality \eqref{eq:push-partial-0} follows from equality \eqref{eq:pull-X}
and the identity $\Psi_* \partial_t = \partial_t + S$, so it 
only remains to establish equation \eqref{eq:pull-dV}.
Recall that $dV_{h}$ is defined by
$N dV_{h} = \partial_t\interior dV_{g}$
and satisfies $N\, dt\wedge dV_h = dV_g$
Hence
\begin{align*}
\hat N\, dV_{\hat h} = \partial_t\interior dV_{\hat g} &= \Psi^*( (\partial_t+ S) \interior dV_{g})\\
&= \Psi^*( (\partial_t+ S)\interior (N dt\wedge dV_h)) \\
&=\Psi^* N\; \Psi^*( dV_h - dt\wedge (S\interior dV_h))\\
&= \Psi^* N \left[ \Psi^*( dV_h ) - dt \wedge\Psi^*(S\interior dV_h)\right].
\end{align*}
Applying $j_*$ to both sides and using
the fact that $\Psi^*N = \hat N$ we find
$j_* dV_{\hat h} = j_* \Psi^* dV_h$.
But $dV_{\hat h}$ is spatial and 
equation \eqref{eq:pull-dV} therefore follows.
\end{proof}

The following key proposition describes the very limited way in which 
the slice Lagrangian of
a diffeomorphism invariant Lagrangian can depend on the shift.
\begin{proposition}\label{prop:X-by-push-only}
Let $F$ be a manifold and let $\mathcal L$ be a minimally 
coupled diffeomorphism invariant Lagrangian with field values in $F$.  Let $g$
be a slice compatible Lorentzian metric on $M$ with
$(n+1)$ decomposition $(h,N,X)$ and let $\fieldsec:M\to F$.
Then
\[
L_\Sigma(\fieldsec,\partial_x \fieldsec, \partial_t \fieldsec, h, N, X) 
= L_\Sigma(\fieldsec,\partial_x \fieldsec, \partial_t\fieldsec - \fieldsec_* X,h,N,0).
\]
\end{proposition}
\begin{proof}
Fix some $p\in M$ and 
let $\Psi$ be a diffeomorphism from a neighborhood $U$ 
of $p$ to a neighborhood $V$ of $p$ that fixes the slice
$\Sigma_{t_0}\cap U$ containing $p$ and such that 
$\Psi_* \partial_t = \partial_t -X$.  Such a diffeomorphism can 
be constructed using the integral curves of $\partial_t-X$.

Let $\hat \fieldsec=\fieldsec\circ \Psi$, $\hat g = \Psi^* g$, and let
$(\hat h, \hat N, \hat X)$ be the $n+1$ decomposition of $\hat g$.
We compute on $\Sigma_{t_0}\cap U$
\begin{align*}
L_\Sigma(\hat \fieldsec,\partial_x \hat \fieldsec, \partial_t \hat \fieldsec, \hat h, \hat N, \hat X) &= j_*(\partial_t \interior \mathcal L(\Psi^* \fieldsec,\partial\Psi^* \fieldsec, 
\Psi^* g))\\
&= j_*(\partial_t \interior \Psi^* \mathcal L(\fieldsec,\partial\fieldsec,g))\\
&= j_*\Psi^*\left[ (\partial_t-X)\interior \mathcal L(\fieldsec,\partial\fieldsec,g) \right] \\
&= j_*\Psi^*\left[ L_\Sigma(\fieldsec,\partial_x\fieldsec,\partial_t\fieldsec,h,N,X) -
X\interior \mathcal L(\fieldsec,\partial\fieldsec,g) \right]\\
&= L_\Sigma(\fieldsec,\partial_x\fieldsec,\partial_t\fieldsec,h,N,X)
\end{align*}
where Lemma \ref{lem:t-preserving-diffeo} Parts 
\ref{part:kills-S-interior-omega} and 
\ref{part:spatial-stays-put}
have been applied at the final step.

Since $\Psi$ is the identity on $\Sigma_{t_0}\cap U$,
we have the following relations on $\Sigma_{t_0}\cap U$:
\begin{itemize}
  \setlength{\itemsep}{0pt}%
  \item $\hat \fieldsec=\fieldsec$ and $\partial_x \hat\fieldsec = \partial_x \fieldsec$
  \item $\hat h = h$ and  $\hat N = N$ 
  (Lemma \ref{lem:pullback-g} and Lemma \ref{lem:t-preserving-diffeo} Part \ref{part:spatial-stays-put}).
\end{itemize}
Moreover, since $\Psi_*\partial_t = \partial_t-X$,
Lemma \ref{lem:pullback-g} equation \eqref{eq:pull-X} implies
$\hat X = 0$ on all of $U$. Hence
\[
L_\Sigma(\fieldsec,\partial_x \fieldsec, \partial_t \hat \fieldsec, h, N, 0)
= L_\Sigma(\fieldsec,\partial_x\fieldsec,\partial_t\fieldsec,h,N,X)
\]
on $\Sigma_{t_0}\cap U$.  Noting that
\[
\partial_t\hat \fieldsec = (\Psi^* \fieldsec)_* \partial_t =
\fieldsec_* \Psi_* (\partial_t) = \fieldsec_* (\partial_t - X)
= \partial_t \fieldsec - \fieldsec_* X
\]
we obtain the desired equality at the arbitrary point $p$
and therefore generally.
\end{proof}

In order to connect the previous result to the momentum
constraint we recall the classical mechanics
notion of momentum conjugate to $\fieldsec$.
Fix a slice $\Sigma_{t_0}$ along with the following
data on $\Sigma_{t_0}$:
\begin{itemize}
  \setlength{\itemsep}{0pt}%
  \item a map $\fieldsec:\Sigma_{t_0}\to F$
  \item a spatial Riemannian metric $h$
  \item a positive function $N$
  \item a spatial vector field $X$.
\end{itemize}
At a point $p\in \Sigma_{t_0}$ we then
obtain a map
$ T_{\fieldsec(p)}(N) \to \Lambda^n_p(M)$
given by
\[
\dotfieldsec \mapsto L_\Sigma(\fieldsec,\partial_x\fieldsec, \dotfieldsec,g,N,X).
\]
The \textbf{momentum conjugate} to $\fieldsec$ at $p$
is the linearization
\[
\Pi_\fieldsec = \frac{\partial L_\Sigma}{\partial \dotfieldsec}
\]
which can be identified with a
$T^*F_{\fieldsec(p)}$-valued spatial $n$-form.  
That is, fixing
an arbitrary choice $\omega$
of a non-vanishing $n$-form on $\Sigma_{t_0}$
(e.g. $\omega = dV_h$, though this is not necessary) we have
\[
\Pi_\fieldsec = \eta\otimes\omega
\]
for some $\eta\in T_{\fieldsec(p)}^* F$. 
If $Z\in T_{\fieldsec(p)} F$ we define
\[
\Pi_\fieldsec(Z) = \eta(Z) \omega,
\]
and the value is evidently independent of the choice
of the pair $(\eta,\omega)$ representing $\Pi_\fieldsec$.

\begin{theorem}\label{thm:form-of-J-manifold-edition}
Using the notation of the preceding discussion,
\begin{equation}\label{eq:J-to-Pi-B-manifold-edition}
\momJ dV_h = -\frac12
\Pi_\fieldsec\circ \fieldsec_*
\end{equation}
\end{theorem}
\begin{proof}
Let $X_s$ be a one-parameter family of shifts at some point $p$
with $X_0=X$ and $\left.\frac{d}{dt}\right|_{s=0} X_s= \delta X$.
Although Proposition \ref{prop:X-by-push-only} is a
statement about fields defined on $M$, local extension
arguments imply that at $p$
\[
L_{\Sigma}(\fieldsec,\partial_x \fieldsec,\dotfieldsec, h, N,X_s)
=  L_\Sigma(\fieldsec,\partial_x \fieldsec,\dotfieldsec-\fieldsec_* X_s, g, N,0).
\]
Taking a derivative with respect to $s$ at $s=0$ and using
equation \eqref{eq:def-J-first} we find
\[
\momJ(\delta X) dV_h = -\frac{1}{2} \Pi_\fieldsec( \fieldsec_* \delta X).
\]
\end{proof}

\section{Lagrangians depending on tensor fields}

The aim of this section is prove a generalization of 
Theorem \ref{thm:form-of-J-manifold-edition} to
matter fields represented by sections of tensor bundles.
The work is more involved because the fields transform
nontrivially under diffeomorphism, and the first step
is a careful decomposition into spatial tensors.

\label{sec:tensor-valued}
\subsection{$(n+1)$ decomposition of tensors}
\label{sec:decomp}
There is more than one way to decompose a tensor on $M$ 
into spatial tensors representing its components. 
One could use, e.g., the coordinate
representation of the tensor in terms of $dt$ and $\partial_t$.
Alternatively,  with a Lorentzian metric $g$ in hand
one could use the 
slice normal $\nu$ in place of $\partial_t$.  The normal
$\nu$ depends on both the lapse $N$ and shift $X$, however,
and it turns out to be convenient for our purposes 
to not involve the lapse
in the tensor decomposition.  Thus we describe below a 
decomposition
based on $\partial_t-X$ rather than the full normal $\nu$.

Let $\mathcal P_X$ be the $g$-orthogonal
projection of tangent vectors to vectors tangential to 
slices $\Sigma_t$.  If $Z=a\partial_t + W$ where $W$
is spatial, then 
\[
Z=a N\nu + aX + W
\]
and $\mathcal P_X(Z) = aX+W$.  The adjoint
$\mathcal P_X^*$ acts on 1-forms $\omega$ by declaring
$(\mathcal P_X^*\omega)(X) = \omega(\mathcal P_X Z)$ for arbitrary
vectors $Z$.  Note that the notation $\mathcal P_X$ is justified
since, in terms of the $n+1$ decomposition $(h,N,X)$,
the projection depends only on the shift.

Given an arbitrary tensor in $T^{k_1}_{k_2}(M)$ 
of rank $k=k_1+k_2$, and given 
a slice-compatible Lorentzian metric $g$
with $n+1$ decomposition $(h,N,X)$
we obtain
$2^k$ spatial tensors of varying rank less than or equal
to $k$ by inserting one of 
$\partial_t-X$ or $j_*$ for each contravariant argument 
and one of $dt$ or $\mathcal P_X^*$ for each covariant
argument. For example, a 1-form $\omega =  b dt + \eta$
where $\eta$ is spatial decomposes to
\[
(b-\eta(X), \eta)
\]
and a vector $Z=a\partial_t+W$ where $W$ is spatial decomposes
to 
\[
(a, W+aX).
\]
Fixing some arbitrary ordering for the resulting tensors (e.g., lexicographically based on the interior product/projection choice made
for each argument) we write $\mathcal S(T;X)$ for the 
collection of $2^k$ spatial tensors obtained this way and call the result the $n+1$ \textbf{decomposition} of $T$ 
determined by the time function $t$ and metric $g$, noting
that the metric is only involved via the shift $X$. 

A tuple of tensors
$B$ representing a possible value of $\mathcal S$ is a
\textbf{decomposed tensor} (of \textbf{unified} type $T^{k_1}_{k_2}$),
and the set of all decomposed tensors is a fiber bundle over $M$
consisting of a Whitney sum of spatial subbundles of tensor bundles.

From the explicit representations above 
it is clear that
$\mathcal S(\cdot;X)$ is invertible for 1-forms and vector fields,
and a proof by induction shows that it is invertible for arbitrary
tensor fields.  Writing $\mathcal S^{-1}(\cdot;X)$ for the recomposition
operation, $\mathcal S^{-1}(B;X)$ can be written in terms of linear
combinations of tensors obtained from $B$ and $X$ 
using the following operations:
\begin{enumerate}
  \setlength{\itemsep}{0pt}%
  \item interior products with $X$,
  \item tensor products with $X$, $dt$ and $\partial_t$,
  \item argument reordering.
\end{enumerate}
In particular, $\mathcal S^{-1}(B;X)$ depends smoothly on $B$.

The $(n+1)$ decomposition map $\mathcal S$ is
natural under $t$-preserving diffeomorphisms
in the following sense.

\begin{lemma}\label{lem:pullback-split} Let $\Psi$ be a $t$-preserving diffeomorphism
from a neighborhood $U$ to a neighborhood $V$ of $M$
and let $g$ be a slice-compatible Lorentzian metric
on $V$.  Given a tensor field $T$
on $V$ define
\begin{align*}
B &= \mathcal S(T;X)\\
\hat B &= \mathcal S(\Psi^* T; \hat X)
\end{align*}
where $\hat X$ is the shift vector of 
the slice-compatible Lorentzian metric 
$\Psi^* g$.
Then 
\begin{equation}\label{eq:pull-back-B}
\hat B = j_* \Psi^* B.
\end{equation}

If additionally $\Psi|_{\Sigma_{t_0}\cap U}=\Id$ 
for some $t_0$, then the following equations hold
on $\Sigma_{t_0}\cap U$:
\begin{align}
\hat B &= B\label{eq:hat-B-unmoved}\\
\Lie_S \hat B &= \Lie_S B\label{eq:Lie-hat-or-not}
\end{align}
for all spatial vector fields $S$.
\end{lemma}
\begin{proof}
To establish equation \eqref{eq:pull-back-B} we demonstrate
below that for 1-forms $\omega$,
\begin{align}
(\partial_t-\hat X)\interior \Psi^*\omega = \Psi^*((\partial_t- X)\interior\omega)\label{eq:pullback-omega-interior}\\
(\Psi^* \omega)\circ j_* = j_* 
\Psi^* (\omega\circ j_*) \label{eq:pullback-omega-restricted}
\end{align}
and that for vector fields $Z$
\begin{align}
dt \interior \Psi^* Z = \Psi^*( dt\interior Z)\label{eq:pullback-Z-interior}\\
(\Psi^* Z) \circ \mathcal P_{\hat X}^* = \Psi^* (Z\circ \mathcal P_{X}^*)
\label{eq:pullback-Z-projected}.
\end{align}
Together these formulas nearly yield $\hat B = \Psi^* B$
for 1-forms and vectors, but equation \eqref{eq:pullback-omega-restricted} leads to the weaker statement
$\hat B = j_* \Psi^*B$.
The same arguments used to establish equations
\eqref{eq:pullback-omega-interior}--\eqref{eq:pullback-Z-projected}
apply to each component of a tensor and lead to the general result.

Equation \eqref{eq:pullback-omega-interior} follows from 
the identity 
$
\partial_t-\hat X = \Psi^*(\partial_t-X)
$
from Lemma \ref{lem:pullback-g}, equation
\eqref{eq:pullback-omega-restricted} is immediate
from Lemma \ref{lem:t-preserving-diffeo} part \ref{part:covariant-almost-commutes}, and 
equation  \eqref{eq:pullback-Z-interior} is a consequence
of the identity $\Psi^* dt = dt$.
Turning to equation \eqref{eq:pullback-Z-projected},
consider a vector $Z$ and let $W=\mathcal P_X(Z)$.
Then
$
Z = a(\partial_t -X) + W
$
for some number $a$ and
$
\Psi^*(Z) = a \Psi^*(\partial_t-X) + \Psi^*W = a(\partial_t-\hat X) +\Psi^*W.
$
Since $\Psi$ preserves $t$, $\Psi^* W$ is spatial as well and 
consequently
\[
\mathcal P_{\hat X} \Psi^* Z = \Psi^* W = \Psi^*(\mathcal P_X Z).
\]
Equation \eqref{eq:pullback-Z-projected} then follows from the definition of the adjoint.

Finally suppose $\Psi=\Id$ on $U\cap \Sigma_{t_0}$ for some $t_0$.
Equation \eqref{eq:hat-B-unmoved} on $U\cap \Sigma_{t_0}$ 
follows from  equation \eqref{eq:pull-back-B} and
Lemma \ref{lem:t-preserving-diffeo}
part \ref{part:spatial-stays-put}.  But then
Lemma \ref{lem:basic-split-Lie} part \ref{part:lie-of-zero}
implies equation \eqref{eq:Lie-hat-or-not}
on $\Sigma_{t_0}\cap U$ as well.
\end{proof}

\subsection{Structure of the momentum density}
Let $E$ be a tensor
bundle $T^{k_1}_{k_2}(M)$.  
A \textbf{minimally coupled Lagrangian depending 
on a field with values} $E$ is a smooth bundle map
\[
\mathcal L = E\oplus \underbrace{(E \oplus \cdots \oplus E)}_{\text{$n+1$ times}} \oplus G(M) \to \Lambda^{n+1}(M)
\]
where $\oplus$ denotes the Whitney sum
and $G(M)$ is the bundle of Lorentzian metrics.
A section $T$ of $E$ and Lorentzian metric $g$ then
determine the $n+1$ form
\[
\mathcal L(T,\partial T, g)
\]
where $\partial T$ is shorthand for the tuple
of Lie derivatives
$(\partial_{x^i} T, \ldots, \partial_{x^n} T, \partial_t T)$.

We say that $\mathcal L$ is \textbf{diffeomorphism invariant} if
whenever $\Psi$ is an orientation preserving 
diffeomorphism from a neighborhood
$U_1$ of $M$ to a neighborhood $U_2$ of $M$,
\[
\Psi^* \mathcal L(T, \partial T, g)
= \mathcal L( \Psi^* T, 
\partial (\Psi^* T), \Psi^* g)\,\;\text{on $U_1$}
\]
whenever $T$ and $g$ are defined on $U_2$.

To describe the slice Lagrangian in this context,
consider a point $p\in M$ and the following data at $p$:
\begin{itemize}
  \setlength{\itemsep}{0pt}
  \item $B$, a decomposed tensor of unified type $T^{k_1}_{k_2}(M)$,
  \item $B_{(i)}$, $i=1,\ldots,n$, decomposed tensors 
  of the same unified type as $B$,
  \item $\dot B$, a decomposed tensor of the same unified type as $B$
  \item $h$, a spatial Riemannian
metric,
  \item $N>0$,
  \item vectors $X$, $X_{(1)},\ldots, X_{(n)}$ and $\dot X$.
\end{itemize}
Let $\tilde B$ be a smooth extension
of $B$ near $p$ with $\partial_{x^i} B = B_{(i)}$ for each $i$
and $\partial_t B = \dot B$ at $p$; an easy coordinate argument
shows that such an extension exists.  Let $\tilde X$ be a similar
extension of $X$.  The associated value of the 
\textbf{slice Lagrangian} $L_{\Sigma}$ at $p$
is
\[
L_{\Sigma}(B,(B_{(1)},\ldots,B_{(n)}),\dot B,h,N,X,
(X_{(1)},\ldots,X_{(n)}),\dot X) = 
\partial_t\interior \mathcal L(\mathcal S^{-1}(\tilde B,\tilde X),\partial \mathcal S^{-1}(\tilde B,\tilde X), g(h,N,\tilde X)).
\]
This map is well defined because the derivatives of $S^{-1}(\tilde B;\tilde X)$ at $p$ are computable solely from the values and
derivatives of $\tilde B$ and $\tilde X$ at $p$ and hence
are independent of the choice of extension.
Because the 
derivatives of $\mathcal S^{-1}(\tilde B,\tilde X)$ at $p$
depend smoothly on the derivatives of $\tilde B$ and $\tilde X$ at $p$,
$L_\Sigma$ is a smooth bundle map taking on values in the subbundle
of spatial $n$-forms.  If $B$ and $X$ are defined in a neighborhood of $p$ then they serve as their own extensions and we find
\[
L_{\Sigma}(B,\partial_x B, \partial_t B,h,N,X,
\partial_x X,\partial_t X) =
\partial_t\interior \mathcal L(\mathcal S^{-1}(B,X),\partial \mathcal S^{-1}(B,X), g(h,N,X))
\]
at $p$ and indeed on the whole neighborhood. 
  
We have the following analog of Proposition \ref{prop:X-by-push-only}.
\begin{proposition}\label{prop:X-by-Lie-only}
Let $E$ be a tensor bundle $T^{k_1}_{k_2}(M)$
and suppose $\mathcal L$ is a minimally coupled
diffeomorphism invariant 
Lagrangian with field values from $E$.
Let $g$ be a slice-compatible metric on $M$ with 
$(n+1)$ decomposition $(h,N,X)$, let $T$
be a section of $E$, and let $B=S(T;X)$.  
Then
\[
L_{\Sigma}(B,\partial_x B, \partial_t B, h, N, X, \partial_x X, \partial_t X) = 
L_{\Sigma}(B,\partial_x B, \partial_t B-\Lie_X B, h, N, 0,  0,  0).
\]
\end{proposition}
\begin{proof}
Let $T$ and $g$ be given and consider some point $p$
in $M$.  As in the proof of 
Proposition \ref{prop:X-by-push-only}, 
let $\Psi$ be a diffeomorphism from a neighborhood $U$ 
of $p$ to a neighborhood $V$ of $p$ that fixes the slice
$\Sigma_{t_0}\cap U$ containing $p$ and such that 
$\Psi_* \partial_t = \partial_t -X$.

Let $\hat g = \Psi^* g$ with $(n+1)$ decomposition
$(\hat h, \hat N, \hat X)$ and let 
$\hat B = \mathcal S( \Psi^* T;\hat X)$. Using diffeomorphism
invariance we compute on $\Sigma_{t_0}\cap U$
\begin{align*}
L_{\Sigma}(\hat B, \partial_x \hat B, \partial_t \hat B,
\hat h, \hat N, \hat X, \partial_x \hat X, \partial_t \hat X)
& =
j_*(\partial_t\interior  \mathcal L( \Psi^*T,\partial \Psi^*T, \hat g))\\
& =
j_*(\partial_t\interior \Psi^* \mathcal L( T,\partial T, g))\\
&= 
j_* \Psi^*\left[(\partial_t-X)\interior
\mathcal L( T,\partial T, g)\right]\\
&= j_*\Psi^* L_{\Sigma}( B, \partial_x B, \partial_t B,
h, N, X, \partial_x X, \partial_t X)
- j_*\Psi^* X\interior \mathcal L( T,\partial T, g)\\
&= L_{\Sigma}( B, \partial_x B, \partial_t B,
h, N, X, \partial_x X, \partial_t X).
\end{align*}
The first equality in the chain uses the fact that $L_\Sigma$ is spatial
(and hence $j_*L_\Sigma=L_\Sigma$); the final equality follows from 
Lemma \ref{lem:t-preserving-diffeo}
Parts \ref{part:kills-S-interior-omega}
 and \ref{part:spatial-stays-put}.

The same argument as in Proposition \ref{prop:X-by-push-only}
shows $\hat h=h$ and $\hat N=N$ on $\Sigma_{t_0}\cap U$
and that $\hat X=0$ on all of $U$.
From the final conclusions of
Lemma \ref{lem:pullback-split} we
find $\hat B = B$ and $\partial_x \hat B=\partial_x B$
on $\Sigma_{t_0}\cap U$, and combining these
observations we conclude that on $U$,
\[
L_{\Sigma}(B, \partial_x B, \partial_t \hat B,
h, N, 0, \partial_x 0, \partial_t 0)
=
L_{\Sigma}(B, \partial_x B, \partial_t  B,
h, N, X, \partial_x X, \partial_t X).
\]

We wish to rewrite 
the remaining hatted term, $\partial_t \hat B$,
in terms of unhatted quantities.
Lemma \ref{lem:pullback-split}, 
Lemma \ref{lem:basic-split-Lie} part \ref{part:lie-of-restrict}
and Lemma \ref{lem:pullback-g} equation \eqref{eq:push-partial-0}
along with the fact that $\hat X\equiv 0$ imply
\[
\partial_t \hat B = \partial_t\; j_* \Psi^* B
= 
j_* \partial_t  \Psi^* B
= j_* \Psi^*( (\partial_t - \Lie_X) B) ).
\]
Since $B$ is spatial, $\partial_t B$ and
$\Lie_X B$ are both spatial (Lemma \ref{lem:basic-split-Lie}
Part \ref{part:lie-of-spatial}).
Hence Lemma \ref{lem:t-preserving-diffeo} 
part \ref{part:spatial-stays-put}
implies
\[
j_*  \Psi^*( (\partial_t - \Lie_X) B ) 
= (\partial_t - \Lie_X) B
\]
on $\Sigma_{t_0}\cap U$.  That is, 
$\partial_t \hat B= (\partial_t - \Lie_X) B$
at $p$.
\end{proof}

As a consequence of Proposition \ref{prop:X-by-Lie-only},
we drop the dependence of $L_\Sigma$ on $\dot X$.

We would like to generalize 
Theorem \ref{thm:form-of-J-manifold-edition}, 
and to do this we
first recall the notion of momentum in this context.
Fix $B$, $h$, $N$, and $X$, on a slice $\Sigma_{t_0}$.
Since $X$ is spatial, the Lie derivatives $\partial_x X$
are well-defined on $\Sigma_{t_0}$ and
we have at each $p$ on $\Sigma_{t_0}$ a smooth map
\[
\dot B \mapsto 
L_{\Sigma}(B,\partial_x B, \dot B, B, h, N,X,\partial_x X) NdV_h
\]
with linearization
\[
\Pi_B := \frac{\partial L_{\Sigma}}{\partial \dot B}.
\]
This linearization is, by definition, the \textbf{momentum conjugate}
to $B$. Fix a nonvanishing spatial $n$-form $\omega$ on $\Sigma_{t_0}$.
Because $B$ is a tuple $(B_1,\ldots,B_{2^k})$
of tensor fields, 
$\Pi_B$ can be represented as a tuple $(B_1^* \omega,\ldots,B_{2^k}^*\omega)$
of tensor field valued $n$-forms, with the starred tensor fields  dual to those of $B$. If $\delta B=(\delta B_1,\ldots,\delta B_{2^k})$ 
is a decomposed tensor of the  same unified type as $B$ we define
\[
\Pi_B(\delta B) = \sum_{j=1}^{2^k} B_j^*(\delta B_j) \omega
\]
with $B_j^*(\delta B_j)$ denoting tensor contraction.
From the close coupling of $\partial_t B$ and $\Lie_X B$ found
in Proposition \ref{prop:X-by-Lie-only} we have
the following immediate conclusion.

\begin{corollary}\label{cor:vary-X-to-Lie}
On $\Sigma_{t_0}$ let $B$ and $\dot B$ be decomposed tensor fields
and let $(h,N,X)$ be a decomposed metric.
Suppose $X_s$ is a one-parameter family of shifts
with $X_0=X$ and define 
$\delta X = \left.\frac{d}{ds}\right|_{s=0} X_s$.  
Then
\[
\left.\frac{d}{ds}\right|_{s=0}
L_{\Sigma}(B,\partial_x B, \dot B, h, N, X_s, \partial_x X_s )
= -\Pi_B ( \Lie_{\delta X} B ).
\]
\end{corollary}

Corollary \ref{cor:vary-X-to-Lie} is unfortunately 
not an immediate statement about the momentum constraint.
The momentum constraint arises from varying the spacetime Lagrangian
with respect to the shift, leaving $T$, $h$ and $N$ fixed.
But the variation appearing in Corollary \ref{cor:vary-X-to-Lie}
leaves $B$ fixed, and $T=\mathcal S^{-1}(B;X)$ 
changes as $X$ changes. In order to apply 
Corollary \ref{cor:vary-X-to-Lie} to the momentum constraint 
we restrict our attention to fields $T$ that are stationary
for $\mathcal L$.  That is, we assume that
if $T_s$ is a path of tensor
fields with $T_0 = T$ and 
$\delta T:= \left.\frac{d}{ds}\right|_{s=0} T_s$ compactly 
supported,
\[
\int_M \left.\frac{d}{ds}\right|_{s=0} \mathcal L(T_s,\partial T_s,g) = 0.
\]

\begin{theorem}\label{thm:form-of-J-tensor-edition}
Using the notation and hypotheses of Proposition \ref{prop:X-by-Lie-only},
suppose that $\mathcal L$ is stationary at a tensor field $T$.
Then for any compactly supported spatial vector field $\delta X$
on a slice $\Sigma_{t}$
\begin{equation}\label{eq:J-to-Pi-B-tensor-edition}
\int_{\Sigma_t} \mathcal{J} ( \delta X) \, dV_h = -\frac{1}{2} 
\int_{\Sigma_t} \,\Pi_B( \Lie_{\delta X}).
\end{equation}
\end{theorem}
\begin{proof}
Let $X_s$ be a path of spatial vector fields on all of $M$
with $\left.\frac{d}{ds}\right|_{s=0} X_s = \delta X$
and let $B=\mathcal S(T;X)$.  To simplify notation we write
\begin{align*}
\frac{\partial L_{\Sigma}}{\partial X}[\delta X] &= 
\left.\frac{d}{ds}\right|_{s=0} L_\Sigma(B,\partial_x B, \partial_t B, h, N, X_s, \partial_x X_s)\\
\frac{\partial \mathcal L }{\partial X}[\delta X]
&= \left.\frac{d}{ds}\right|_{s=0} \mathcal L(T,\partial T, g(h,N,X_s))\\
\frac{\partial \mathcal S^{-1}(B;X) }{\partial X}[\delta X]
&= \left.\frac{d}{ds}\right|_{s=0} \mathcal S^{-1}(B;X_s).
\end{align*}
Note that $\mathcal S(B;X_s)$ determines a path $T_s$ of tensor
fields.
Given an arbitrary such path we set
$\delta T = \left.\frac{d}{ds}\right|_{s=0} T_s$
and write
\[
\frac{\partial \mathcal L }{\partial T}[\delta T]
= \left.\frac{d}{ds}\right|_{s=0} \mathcal L(T_s,\partial T_s, g).
\]

Equation \eqref{eq:J-to-Pi-B-tensor-edition} is established
by computing $\int_M dt\wedge \frac{\partial L_{\sigma}}{\partial X}$
two different ways.  First, pointwise
\begin{align*}
dt\wedge \frac{\partial L_{\sigma}}{\partial X}
&=\left.\frac{d}{ds}\right|_{s=0}
\mathcal L( S^{-1}(B;X_s),\partial S^{-1}(B;X_s) ,g(h,N,X_s))\\
&=\int_{M} \left[ \frac{ \partial \mathcal L}{\partial T}\left[\frac{ \partial S^{-1}(B;X)}{\partial X} [\delta X]\right]
+ \frac{ \partial \mathcal L}{\delta X}[\delta X] \right].
\end{align*}
Integrating over $M$ and using
equation \eqref{eq:def-J-first} and
 the fact that $\mathcal L$
is stationary at $T$ we conclude
\begin{align*}
\int_M dt\wedge \frac{\partial L_{\sigma}}{\partial X} &=
\int_M \left.\frac{d}{ds}\right|_{s=0}
\mathcal L( S^{-1}(B;X_s),\partial S^{-1}(B;X_s) ,g(h,N,X_s))\\
&= \int_M \frac{ \partial \mathcal L}{\delta X}[\delta X]\\
&=2\int_M \momJ(\delta X)dt\wedge dV_h.
\end{align*}
On the other hand, from Corollary \ref{cor:vary-X-to-Lie}
we have
\[
\int_{M} dt\wedge \frac{ \partial L_\Sigma}{\partial X}[\delta X]
=-\int_{M} dt\wedge \Pi_B(\Lie_{\delta X} B)
\]
and hence
\begin{equation}\label{eq:J-to-Pi-on-M}
-2\int_M \momJ(\delta X)dt\wedge dV_h=
\int_{M} dt\wedge \Pi_B(\Lie_{\delta X} B).
\end{equation}
This integration is on all of $M$, whereas equation
\eqref{eq:J-to-Pi-B-tensor-edition} involves integration on a single
slice.  Nevertheless, since $\delta X$ is spatial, 
for an arbitrary function $f(t)$, $\Lie_{f(t)\delta X} B=f(t)\Lie_{\delta X} B$, and equation \eqref{eq:J-to-Pi-B-tensor-edition}
therefore follows from equation \eqref{eq:J-to-Pi-on-M}
and a concentration argument.
\end{proof}

Except in rare circumstances, 
the integrand on the right-hand side of equation \eqref{eq:J-to-Pi-B-tensor-edition} is not a tensorial
expression in $\delta X$ and we cannot claim
$\momJ(\delta X)\, dV_h = -\frac12 \Pi_B( \Lie_{\delta X} B)$
pointwise.
Nevertheless, $\momJ\, dV_h$ is completely determined by
$\Pi_B$ and $B$.
Indeed, in local coordinates
one can write an expression for $\momJ dV_h$
in terms of $\Pi_B$ and $B$ using integration by parts
to remove derivatives from $\delta X$ that appear
in $\Lie_{\delta X} B$.  Although this operation
is done in coordinates, the result is a manifestly
natural operation, though it is not one that is familiar to us
in its full generality.  Moreover, the exact expression
appearing in the momentum constraint can involve
rewriting terms appearing after integration by parts
using constraints satisfied by the matter fields.
See Sections \ref{secsec:proca} and
\ref{secsec:EMCD} where we carry out this process 
explicitly for concrete examples.


\section{Electromagnetism-Charged-Scalar Field (EMCSF)}
\label{sec:EMCSF}
There are subtleties in the results of the previous section,
and it is helpful to illustrate them in the context of a 
concrete matter model: electromagnetism coupled to a charged scalar field (EMCSF).

Consider a trivial $\mathbb C$ bundle over $M$
with structure group $U(1)$ and let $\mathcal D$
be a connection on the bundle.  By parallel transport
we can choose a
common length scale on the fibers of the bundle and subsequently
select a global section $f$  over $M$ with $|f|=1$
on each fiber. 
The connection $\mathcal D$ can be represented
by means of a 1-form $\mathcal A$ via
\[
\mathcal D f = i \chargec \mathcal A f,
\]
where the constant $\chargec$ is the charge of the scalar field.

Once the choice of frame $f$ is made an arbitrary section 
$\mathcal S$ of the bundle can be represented by means of a section $z$
of the trivial bundle $\mathbb{C}\times M$:
\[
\mathcal S = z f.
\]
Under a global change of frame $\tilde f = e^{i\chargec \Xi} f$
we have the transformations
\begin{equation}\label{eq:EMCSF-gauge-trans}
\begin{aligned}
\tilde z &= e^{-i\chargec \Xi} z\\
\tilde {\mathcal A} &= {\mathcal A} + d\Xi.
\end{aligned}
\end{equation}

Let $g$ be a fixed (slice compatible) background Lorentzian metric 
on $M$ with $(n+1)$ decomposition $(h,N,X)$.
Leaving the background frame $f$ for $E$ implicit, the matter fields
for EMCSF consist of a 1-form $\mathcal A$
and a complex field $z$, and the EMCSF Lagrangian is
\begin{equation}\label{eq:Lag-EMCSF}
\mathcal L_{\rm EMCSF} = -\left[ \frac{1}{8\pi} |d\mathcal A|_g^2 + |dz+i\chargec z\mathcal A|_{\mathbb{C},g}^2\right]dV_g.
\end{equation}
The Lagrangian is easily seen to be diffeomorphism invariant, and is 
also invariant under 
gauge transformations
of the form \eqref{eq:EMCSF-gauge-trans}.  It depends on tensor-valued fields ($\mathcal A$ and $f$)
and hence the EMCSF model falls in the category of
matter fields, treated in Section \ref{sec:tensor-valued}.

\subsection{Illustration of Proposition \ref{prop:X-by-Lie-only}}
Proposition \ref{prop:X-by-Lie-only} concerns the structure
of the Lagrangian after the matter fields have
carefully split into spatial components.  Following
the procedure of Section \ref{sec:decomp}, the scalar field
$z$ is left untouched but we decompose
$\mathcal A$ as follows:
\begin{align}
A &= j_* \mathcal A \\
A_{\vdash} &= (\partial_t - X)\interior \mathcal A.
\end{align}
Note that $(\partial_t -X) = N \nu$ where $\nu$ is the slice
unit normal, but the use of $A_\vdash$ rather than,
e.g., $A_\perp:=\nu\interior \mathcal A$ or $A_0:=\partial_t \interior \mathcal A$, 
was and essential tool in the proof of 
Proposition \ref{prop:X-by-Lie-only}.

Proposition \ref{prop:X-by-Lie-only} asserts that when
the slice Lagrangian $L_\Sigma  = \partial_t \interior
\mathcal L_{\rm EMCSF}$ is written in terms of the spatial metric $h$,
the lapse $X$, the shift $N$, and the spatial variables $z$,
$A_\vdash$ and $A$,  then the shift $X$ and the
Lie time derivatives $\dot A = \partial_t A$, $\dot A_\vdash = \partial_t A_\vdash$
and $\dot z = \partial_t z$ appear
in the slice Lagrangian only in the following combinations:
\begin{itemize}
  \item $\dot A - \Lie_X A$,
  \item $\dot A_{\vdash} - \Lie_X A_{\vdash}$,
  \item $\dot z - \Lie_X z = \dot z - X(z)$.
\end{itemize}
Additionally, the lapse appears only algebraically in the slice
Lagrangian; its spatial and time derivatives are absent. 
We now show by direct computation that these assertions
indeed hold for the EMCSF model.



For convenience define the electric covector field
$E$ and its analog $P$ for the scalar field via
\begin{align*}
E &= -j_* ( \nu \interior d\mathcal A)\\
P &= \nu \interior (dz + i\chargec z\mathcal A).
\end{align*}
Using the notation $\mathbf d = j_* d$
and the identities $\partial_t = N\nu+X$
and
$\Lie_X A = \mathbf d(X \interior A) + X\interior \mathbf d A$,
we find
\begin{align}
E &= -\frac{1}{N}\left[ \dot A - X\interior \mathbf {d} A - \mathbf d(X\interior A) - \mathbf d A_\vdash \right]\nonumber\\
&= -\frac{1}{N}\left[ \dot A - \Lie_X A - \mathbf d A_{\vdash}\right]
\label{eq:E-via-A}
\end{align}
and
\begin{equation}\label{eq:P-def}
P = \frac{1}{N} \left[ \dot z - X(z) + i\chargec z A_\vdash\right].
\end{equation}
The slice Lagrangian then becomes
\begin{equation}\label{eq:EMCSF-LSigma}
L_{\Sigma} = \frac14 \left[ 2|E|_h^2 - |\mathbf d A|_h^2 + 2|P|_{\mathbb{C},h}^2 -2|\mathbf d z + i z A|_{\mathbb C,h}^2\right] N dV_h.
\end{equation}
Noting that $E$ and $P$ in equation \eqref{eq:EMCSF-LSigma}
are merely shorthand for
the full expressions \eqref{eq:E-via-A} and \eqref{eq:P-def}
we observe that the shift $X$ and the time derivatives
$\dot A$, $\dot A_\vdash$ and $\dot z$ indeed appear in the Lagrangian 
only 
in the combinations $\dot A - \Lie_X A$ and $
\dot z - X(z) = \dot z - \Lie_X z$.
The combination $\dot A_{\vdash}-\Lie_X A_{\vdash}$ would also
have been allowed by Proposition \ref{prop:X-by-Lie-only},
but the Lagrangian does not depend on $\dot A_{\vdash}$, a
reflection of the previously mentioned gauge freedom.  Finally we observe
that although the slice Lagrangian depends on derivatives of $X$
via the Lie derivatives, it depends only algebraically on $N$.  
Thus we have illustrated Proposition \ref{prop:X-by-Lie-only} in 
this special case.

\subsection{Illustration of Theorem \ref{thm:form-of-J-tensor-edition}}

Conjugate momenta arise from varying the slice
Lagrangian $L_\Sigma$ with respect to the time derivatives
of the field variables, whereas the momentum density arises
from a variation of the slice Lagrangian with respect 
to the shift via equation \eqref{eq:def-J-first}.  
The structural form of the slice Lagrangian implied
by Proposition \ref{prop:X-by-Lie-only}, wherein
time derivatives of field variables appear only in
certain combinations with the shift, is the key
tool needed to relate conjugate momenta to
the momentum density.  This specific relationship is the content
of Theorem \ref{thm:form-of-J-tensor-edition},
which asserts for the EMCSF Lagrangian 
that the integral relationship
\begin{equation}\label{eq:EMCSF-JdV}
-2 \int_{\Sigma_t} \momJ(\delta X) dV_h 
= \int_{\Sigma_t} \left[\Pi_{A}(\Lie_{\delta X}A)+ 
\Pi_{A_\vdash}(\Lie_{\delta X}A_{\vdash}) +
\Pi_{z}(\Lie_{\delta X} z )\right]
\end{equation}
holds for any compactly supported vector field 
$\delta X$ on a slice $\Sigma_t$.   In particular, 
if the conjugate momenta $\Pi_A$, $\Pi_{A_\vdash}$ and $\Pi_z$
are known, then so is $\momJ dV_h$, a manifestation of 
Meta-Theorem \ref{mt:meta-theorem}.

There is an important caveat, however.
Theorem \ref{thm:form-of-J-tensor-edition} assumes that
we are working with fields for which the Lagrangian
is stationary. In particular, any constraints on the
spatial variables implied by the Euler-Lagrange equations
must hold. For the EMCSF fields, gauge freedom leads to the
constraint
\begin{equation}\label{eq:EMCSF-constraint}
\int_{\Sigma_t} \left[-\Pi_{A}(\mathbf d \theta) + \Pi_z( i\chargec z\theta)\right] = 0
\end{equation}
which is essentially Gauss' Law and, as seen below, 
is necessary to establish equation \eqref{eq:EMCSF-JdV}.



We now illustrate Theorem \ref{thm:form-of-J-tensor-edition}
by deriving equation \eqref{eq:EMCSF-JdV} directly.
Recall that the momenta $\Pi_A$, $\Pi_{A_\vdash}$ and
$\Pi_z$ are obtained by varying 
the slice Lagrangian $L_{\Sigma}$ with respect to
$\dot A$, $\dot A_\vdash$ and $\dot z$ respectively.
We find
\begin{align}
\Pi_{A} &= -\frac{1}{2\pi}\left<E,\mathbf d\, \cdot\, \right>_{h} dV_h\label{eq:Pi-A}\\
\Pi_{A_\vdash} &= 0\label{eq:Pi_A_perp}\\
\Pi_{z} &= 2\Re(\overline P\, \cdot \,)  dV_h.\label{eq:Pi-z}
\end{align}
Equation \eqref{eq:def-J-first} implies that $\momJ dV_h$
can be computed by varying $\mathcal L_{\mathrm{EMCSF}}$ with respect
to the shift, and from it we find 
\[
\frac{\partial L_\Sigma}{\partial X}[\delta X] = 2 \momJ(\delta X) dV_h.
\]
Naively, one would like to simply substitute
expressions \eqref{eq:E-via-A} and \eqref{eq:P-def}
for $E$ and $P$ into equation \eqref{eq:EMCSF-LSigma}
and take a derivative with respect to $X$.  Doing so 
we would find that for any compactly supported
spatial vector field $\delta X$
\begin{equation}\label{eq:J-as-Pi-ideal}
-2\momJ(\delta X) dV_h = \Pi_A(\Lie_{\delta X} A) 
+ \Pi_z(\Lie_{\delta X} z)
+\Pi_{A_{\vdash}}(\Lie_{\delta X} A_\vdash),
\end{equation}
which is almost equation \eqref{eq:EMCSF-JdV}.
This cannot be correct, however, because $\Lie_{\delta X} A$ is not
tensorial in $\delta X$.  

To correct this error, and see why
\eqref{eq:EMCSF-JdV} only holds under an integral sign, and only 
under the additional constraint \eqref{eq:EMCSF-constraint}, 
recall that the variation in equation \eqref{eq:def-J-first}
leaves the matter fields $\mathcal A$ and $z$ fixed. 
But we have written $E$ and $P$ in terms of the variable
$A_\vdash = (\partial_t-X)\interior \mathcal A$. 
Although this choice 
is important for obtaining Proposition \ref{prop:X-by-Lie-only},
and for deriving equation \eqref{eq:EMCSF-LSigma}, 
it is not useful to continue using it
here because fixing
$A_\vdash$ while varying $X$ implies that $\mathcal A$
changes.  

For the purpose of computing the variation with respect to $X$,
we introduce $A_0 = \partial_t \interior \mathcal A$,
a variable independent of $X$.  The field $\mathcal A$ 
is determined by $A$, $A_0$ and $\partial_t$  alone.
Since
\[
A_\vdash = A_0 - X\interior \mathcal A = A_0 - X\interior A
\]
and we can rewrite
\begin{align*}
E &= -\frac{1}{N}\left[ \dot A - X\interior \mathbf d A - \mathbf d A_0\right]\\
P &= \frac1N\left[\dot z -X(z) -i\chargec z(X\interior A)+i\chargec z A_0\right].
\end{align*}
The essential point is that this form of $E$ and $P$ is written so that
even if $X$ is changed, the field $\mathcal A$ does not.

We can now vary the Lagrangian \eqref{eq:EMCSF-LSigma} with respect to $X$
while holding $A$ and $A_0$ (and therefore also $\mathcal A$) fixed.
A computation using equation \eqref{eq:def-J-first} 
and equations \eqref{eq:Pi-A}--\eqref{eq:Pi-z}
then implies
\begin{equation}\label{eq:EMCSF-J-true}
-2\momJ(\delta X) dV_h = \Pi_A( \delta X\interior \mathbf d A) + \Pi_z\left( \delta X(z)   + i\chargec z(\delta X\interior A)\right),
\end{equation}
which is the true form of the momentum density.  
It is indeed tensorial in $\delta X$, but 
the clean structure of equation
\eqref{eq:J-as-Pi-ideal} has been lost.  To 
recover it in the integral sense we need to take 
into account a constraint arising from 
equation \eqref{eq:Pi_A_perp}.  

Because the momentum $\Pi_{A_{\vdash}}$ vanishes identically,
a solution of the EMCSF Euler-Lagrange
equations necessarily satisfies on each slice
\[
\int_{\Sigma_t} \frac{\partial L_{\Sigma}}{\partial A_\vdash}[\delta A_\vdash] = 0
\]
for compactly supported functions $\delta A_\vdash$.
As a consequence, a computation leads to
the constraint \eqref{eq:EMCSF-constraint} 
for any compactly 
supported function $\theta$ on a slice $\Sigma_t$.
Integrating equation \eqref{eq:EMCSF-J-true} and applying
the constraint \eqref{eq:EMCSF-constraint} with 
$\theta = \delta X\interior A$, and also
using equation \eqref{eq:Pi_A_perp}, we find
\begin{align}
-2\int_{\Sigma_t} \momJ(\delta X)dV_h
&= \int_{\Sigma_t} \Pi_A(\delta X\interior  \mathbf dA) + \Pi_z(\delta X(z))
+\Pi_z( i\chargec z(\delta X\interior A))\nonumber\\
&=
\int_{\Sigma_t} \Pi_A(\delta X \interior \mathbf dA) + \Pi_z(\delta X(z))
+\Pi_A( \mathbf d( \delta X\interior A) )\nonumber\\
&=\int_{\Sigma_t} \Pi_A(\Lie_{\delta X} A) + \Pi_z(\Lie_{\delta X} z)\nonumber\\
&=\int_{\Sigma_t} \Pi_A(\Lie_{\delta X} A) + \Pi_z(\Lie_{\delta X} z)
+ \Pi_{A_\vdash}(\Lie_{\delta X} A_{\vdash})\label{eq:J-as-Pi-integral}.
\end{align}
Thus we have derived equation \eqref{eq:EMCSF-JdV}.  The
important point is that equation \eqref{eq:J-as-Pi-ideal} only 
holds under an integral sign, and only if the matter fields
satisfy any constraints imposed by the matter Euler-Lagrange equations.

Concerning the constraint \eqref{eq:EMCSF-constraint}, 
Section \ref{secsec:EMCSF-constraint} below outlines
a systematic approach to constructing its solutions.
We show additionally in Section \ref{secsec:ConfMethod}
how core features of the conformal method of solving
the gravitational field constraints are
direct analogs of the EMCSF constraint solving procedure.

\subsection{Field Values and Conformal Changes}
As mentioned at the end of the introduction, 
in applications it may be more convenient to describe
matter in terms of variables different from the ones
used thus far in this section.   For example, 
\[
\Pi_A = -\frac{1}{2\pi}\ip<E,\cdot>_h dV_h = - \frac{1}{2\pi}E^\sharp\tensor dV_h
\]
and hence the electric field expressed either as a covector ($E$)
or a vector ($E^\sharp$) encodes $\Pi_A$, so long as the metric
is also understood.
If $\Pi_A$ remains fixed but $h\mapsto h^* = \phi^{q-2} h$,
we obtain the conformal transformation laws
\begin{align*}
E^* &=  \phi^{-2} E\\ 
(E^*)^\sharp &= \phi^{-q} E^\sharp.
\end{align*}
Similarly, $P$ encodes $\Pi_z$ with the conformal transformation
rule $P^* = \phi^{-q} P$.  On a slice $\Sigma_t$, the 
constraint \eqref{eq:EMCSF-constraint} can be expressed
in metric terms as
\begin{equation}\label{eq:constraint-EMCSF-EP}
\div_h E = -4\pi\chargec \Im(\overline P z),
\end{equation}
where we are using the same notation for spatial fields on $M$
and their their restrictions to $\Sigma_t$.
The fact that the constraint \eqref{eq:EMCSF-constraint}
does not actually involve the metric implies that this
equation must be conformally invariant.  Indeed, 
the invariance of equation \eqref{eq:constraint-EMCSF-EP}
follows from the conformal transformation laws
for $E$ and $P$ along with the rule
$\div_{h^*} \phi^{-2} = \phi^{-q} \div_h$
when acting on covector fields.

\section{The Conformal Method and Phase Space}
\label{sec:constraint-solving}

Theorems \ref{thm:form-of-J-manifold-edition}
and \ref{thm:form-of-J-tensor-edition} concern the
Einstein equations generally and are
not statements about the conformal method. 
As far as the conformal method is concerned, 
we use the following consequence of
these theorems:
when momentum density is written
as a function of the slice metric $h$ along with certain suitably 
decomposed matter fields
$B$ and conjugate momenta $\Pi_B$, then
\begin{equation}\label{eq:J-transform}
\momJ(\phi^{q-2}h,B,\Pi_B) = \phi^{-q} \momJ(h,B,\Pi_B).
\end{equation}
At this point, the hurried reader could 
continue to the applications of Section \ref{sec:apps}
after scanning the equations of the conformal method in its Lagrangian and
Hamiltonian forms
(equations \eqref{eq:Lich-Lag}--\eqref{eq:LCBY-mom-Lag}
and \eqref{eq:Lich-Ham}--\eqref{eq:LCBY-mom-Ham} below)
to verify that, as outlined in the introduction, 
decoupling of the momentum constraint occurs
under the the transformation rule \eqref{eq:J-transform}.

The larger picture of our work, however, concerns the very close
association of the conformal method with the Hamiltonian 
formulation of the evolution problem: initial data for
the matter variables are profitably specified by fields 
and their conjugate momenta,
and (at least for CMC data) an analogous statement is true for
the gravitational variables.  Moreover, the conformal method
itself is a rather direct analog of a natural technique for 
parameterizing solutions of the EMCSF constraint \eqref{eq:EMCSF-constraint} from Section \ref{sec:EMCSF}. 
We therefore outline the EMCSF 
construction in Section \ref{secsec:EMCSF-constraint} 
as a means of motivating the conformal
method, which we summarize in Section \ref{secsec:ConfMethod}.
Section \ref{secsec:CMC} discusses the CMC specialization
of the conformal method, and flags some important 
differences between the vacuum and non-vacuum settings.

\subsection{Solution of the EMCSF constraint}
\label{secsec:EMCSF-constraint}

Recall that the EMCSF field variables are a 1-form $\mathcal A$
and a complex function $z$.
If $(\mathcal A, z)$ is a solution of the Euler-Lagrange
equations for the Lagrangian \eqref{eq:Lag-EMCSF},
then so is $(\mathcal A + d\Xi, e^{-i\chargec\Xi}z)$
for any function $\Xi$; 
it represents the same physical solution but with respect to a 
different frame. If we consider a gauge transformation
with $\Xi\equiv 0$ at some time $t_0$ then 
we find that the spatial variables
$(A,A_\vdash,z)$ corresponding to the spacetime
fields $(\mathcal A,z)$ at $t=t_0$
are unaffected but their
velocities transform as
\begin{align*}
\dot A &\mapsto \dot A + \mathbf d \dot \Xi\\
{\dot A}_\vdash & \mapsto {\dot A}_\vdash + \dot \Xi\\
\dot z &\mapsto \dot z -i\chargec \dot\Xi z.
\end{align*}
Thus, because of gauge freedom,
the state of the physical system is
invariant under velocity transformations
\begin{equation}\label{eq:EMCSF-gauge-v}
(\dot A,{\dot A}_\vdash,\dot z) \mapsto 
(\dot A,{\dot A}_\vdash,\dot z) + (\mathbf d \theta, \theta, -i\chargec z\theta)
\end{equation}
with the arbitrary function $\theta$ playing
the role of $\dot\Xi$.
Field velocities are meaningful only modulo such transformations.
Because $\Pi_{ {A}_\vdash}$ always vanishes for this 
system, the constraint \eqref{eq:EMCSF-constraint} 
at some $t$ can be written
\begin{equation}\label{eq:EMCSF-constraint-full}
\int_{\Sigma_t}\Pi_A( \mathbf d \theta) + \Pi_{A_\vdash}(\theta) +\Pi_z(-i\chargec z\theta) = 0
\end{equation}
and consequently the total system momentum is insensitive to velocity
perturbations of the form \eqref{eq:EMCSF-gauge-v}. 

One approach to building solutions of the constraint is
to specify field values $(A,A_{\vdash},z)$ along with
the system velocity.  Specifically, we provide 
$(\dot A, \dot A_{\vdash},\dot z)$ and the system 
velocity is the equivalence class
\[
(\dot A, \dot A_{\vdash},z)+(\mathbf d\theta,\theta, -i\chargec z\theta)
\]
where $\theta$ is an arbitrary function. Recall that velocity and momentum
are related to each other in classical mechanics by the Legendre
transformation of the system \cite{goldstein_classical_2002}, and the constraint 
\eqref{eq:EMCSF-constraint-full} is really a constraint on the momentum.
Thus the strategy is to write down the momentum in terms
of a velocity with the arbitrary function $\theta$ built into the expression, and then use the constraint to determine 
the choice of $\theta$. The end result is a
system momentum that satisfies the constraint and that is related, via the Legendre transformation, to the specified velocity.

The Legendre transformation of the EMCSF system
is effectively given by equations \eqref{eq:E-via-A}
and \eqref{eq:P-def}.  Indeed, assuming the metric data $(h,N,X)$
is known, then $\dot A$  and $\dot A_{\vdash}$ determine
$E$ and $P$ according to these equations, and $E$ and $P$
determine $\Pi_A$, $\Pi_z$ via equations \eqref{eq:Pi-A} and
\eqref{eq:Pi-z}. Gauge freedom manifests itself in
a degeneracy in the Legendre transformation, and
$\Pi_{A_\vdash}$ is always $0$ and is unaffected
by changes in $\dot A_\vdash$.

Rewriting equations  \eqref{eq:E-via-A}
and \eqref{eq:P-def} and inserting the arbitrary function $\theta$
we find\begin{align*}
\dot A + \mathbf d\theta &= -N E + \Lie_X A + \mathbf d A_\vdash \\
\dot z -i\chargec z\theta &= N P + \Lie_X z -i\chargec z A_\vdash
\end{align*}
Then, introducing 
$\Theta=A_\vdash - \theta$, constraint \eqref{eq:EMCSF-constraint}
in the form of \eqref{eq:constraint-EMCSF-EP} can be written
on $\Sigma_{t}$ as
\begin{equation}\label{eq:EMCSF-constraint-Lag}
-\div_h\left(\frac1N  d \Theta\right) + 4\pi\chargec^2\frac{|z|^2}{N}\Theta
= -\div_h\left(\frac1N (\dot A -\Lie_X A)\right) + \frac{4\pi\chargec}{N} \Im\left(
(\overline {\dot z - \Lie_X z}) z\right).
\end{equation}
As we have done in arriving at \eqref{eq:constraint-EMCSF-EP},
we are using the same notation for spatial tensor fields on $M$
and for their restrictions to $\Sigma_{t}$, and we note
that the `spatial' exterior derivative $\mathbf d$ on $M$ has 
been replaced with the exterior derivative $d$ intrinsic to $\Sigma_{t}$.
Equation \eqref{eq:EMCSF-constraint-Lag} is an elliptic
PDE for $\Theta$ with a favorable sign on the potential 
term $4\pi\chargec^2|z|^2/N$. When supplemented with suitable boundary conditions,
one can find a unique solution $\Theta$ of equation \eqref{eq:EMCSF-constraint-Lag}.  The
remaining momentum  $\Pi_{A_\vdash}$ always vanishes,
and we are therefore able to construct a full set of momenta
from field values $(A,A_{\vdash},z)$ and velocities
$(\dot A,\dot A_{\vdash},\dot z)$ along with the metric data
$(h,N,X)$.  This strategy for solving the EMCSF constraint is
analogous to the Lagrangian formulation of the conformal method
found in the following section.

An alternative, closely related, parameterization involves
starting with momenta $(\tilde \Pi_A,\tilde \Pi_z)$ that do
not satisfy the constraint, but such that $\tilde \Pi_A$
satisfies the vacuum constraint
\[
\int_{\Sigma_{t}} \tilde\Pi_A( d\theta) = 0
\]
for all functions $\theta$ on $\Sigma_{t}$.
We then seek momenta $(\Pi_A,\Pi_z)$ that solve the
constraint and which have a corresponding system velocity (up
to gauge) that is the same as what would be determined by
$(\tilde \Pi_A,\tilde \Pi_z)$ in the vacuum case via the inverse
Legendre transformation.
An argument similar to the one of the previous paragraph
shows that this reduces to seeking 
$(E,P)$ of the form
\begin{align*}
E &= \tilde E -\frac{1}{N} d\theta\\
P &= \tilde P -\frac{1}{N} i\chargec z\theta
\end{align*}
where $(\tilde E, \tilde P)$ are the metric representations
of $(\tilde \Pi_A,\tilde \Pi_z)$, and so $\div_h \tilde E= 0$.
The constraint \eqref{eq:EMCSF-constraint} is satisfied
when $\theta$ solves the PDE
\begin{equation}\label{eq:EMCSF-constraint-Ham}
-\div_h\left(\frac1N  d\theta\right) + 4\pi\chargec^2\frac{|z|^2}{N}\theta
= -4\pi\chargec \Im \left(\overline{\tilde P} z\right).
\end{equation}
This variation of parameterizing solutions of
the EMCSF constraint \eqref{eq:EMCSF-constraint}
corresponds to the Hamiltonian formulation of the
conformal method.

\subsection{Conformal Method}\label{secsec:ConfMethod}

Although we have until this point 
decomposed the spacetime
metric $g$ into spatial variables $(h,N,X)$, 
the conformal method is based on
an alternative,
closely related, 
decomposition $(h,\alpha,X)$ where the spatial
metric $h$ and shift $X$ have the same meaning
as before but where $\alpha$ and the lapse
$N$ are related by
\[
\alpha = dV_h /N.
\]
We call $\alpha$ the \textbf{slice energy density} because
a particle with unit mass and with velocity equal to
the unit surface normal $ \nu = (\partial_t-X)/N$
has energy $1/N$ in the background coordinate system
$(x^i,t)$.
If the metric $h$ conformally transforms to $\phi^{q-2} h$,
and if $\alpha$ remains fixed, then the conformal
transformation law for volume $dV_h \mapsto \phi^q dV_h$ implies
the lapse transforms via $N\mapsto \phi^{q} N$; this
is called a \textbf{densitized lapse}.  We find that there
are conceptual
advantages to working with the conformally
invariant slice energy density $\alpha$ rather
than the more commonly used densitized lapse,
but these objects are equivalent to each other.

Although the reasons for the importance of 
slice energy density remain unclear, a hint can be found in 
 \cite{maxwell_initial_2021} where it is shown
 that the use
of slice energy density as a variable instead of the 
usual lapse
allows for a clean separation
between conformal and volume components
of the kinetic energy part of the 
the ADM gravitational Lagrangian.  This separation
affects the associated Legendre transformation 
in a way that facilitates specifying conformal class 
data (values and momenta) 
independently from volume information.  In the conformal
method, the metric $h$ is then decomposed into two variables,
the conformal class $\mathbf h = [h]$ and the volume form $dV_h$,
with the conformal class becoming the primary gravitational variable,
and the volume form determined by the constraint equations.

In the Lagrangian formulation of the conformal method,
analogous to the EMCSF construction \eqref{eq:EMCSF-constraint-Lag},
we specify the following data:
\begin{itemize}
  \item A conformal class $\mathbf h$
  and a conformal class velocity $\dot {\mathbf h}$ analogous
  to EMCSF data $A$ and $\dot A$,
  \item gauge data $\alpha$ and $X$ analogous to $A_\vdash$,
  \item a mean curvature function $\tau$ with no direct EMCSF analog,
  \item matter field information comparable
  to $z$ and $\Pi_z$ in the EMCSF construction.
\end{itemize}
The anomalous role of the mean curvature in this collection
stands out, and the absence of a classical mechanical interpretation 
is perhaps related to the deficiencies of the
conformal method in the non-CMC setting.  Notably, however,
when $\tau$ is constant it encodes a momentum conjugate
to volume \cite{maxwell_initial_2021}.  

The gauge group for the Einstein gravitational field equations
 consists of spatial diffeomorphisms rather
than EMCSF frame transformations, but the strategy remains the same.
We specify the conformal class velocity $\dot {\mathbf h}$
only up to gauge.  The constraint equations
impose a condition on the conjugate momentum of the conformal class,
and the Legendre transformation of the system connects the
conformal class velocity to the conformal class momentum. The
gauge freedom in the conformal class velocity is then used
to find a conformal class momentum that is related, via the Legendre
transformation, to the conformal class velocity and that
satisfies the constraint.

To make this procedure concrete, and to derive a PDE system
that one
can actually analyze, in practice one specifies the following
seed data:
\begin{itemize}
  \item A metric $h$ determining the conformal class $\mathbf h=[h]$.
  \item A trace-free $(0,2)$ tensor $U$ that, together with $h$,
  encodes the conformal class velocity $\dot {\mathbf h}$ 
  as explained below.
  \item Gauge data $\alpha$ and $X$, from which we define $N = dV_h/\alpha$.
  \item A mean curvature function $\tau$ 
  and matter field information 
as specified for various example cases below.
\end{itemize}
Following \cite{maxwell_conformal_2014}, the conformal class 
velocity is an equivalence class of pairs $(\phi^{q-2}h, \phi^{q-2}U)$
where $\phi$ is an arbitrary conformal factor; the trace-free
condition on $U$ arises because a path of metrics with constant
volume form has a trace-free derivative, and the conformal scaling
on $U$ arises naturally as a consequence of scaling the metric.  
One can think
of $(h,U,X,N,\tau)$ as the gravitational seed data, which specify
the same conformally invariant information as 
$(\phi^{q-2}h, \phi^{q-2}U,\phi^q N, X, \tau)$,
namely
\[
(\mathbf h = [h],\hskip 1ex \dot {\mathbf h} =  [(h,U)],\hskip 1ex \alpha = (1/N)dV_h,\hskip 1ex X,\hskip 1ex \tau).
\]
Diffeomorphism gauge freedom manifests itself in that
rather than specifying a conformal class velocity $[(h,U)]$,
we specify $[(h,U+\ck_h W)]$ where $W$ is an arbitrary vector field
and where $\ck_h W$ is the trace-free part of $\Lie_W h$, namely
\begin{equation}\label{eq:conformal-killing}
(\ck_h W)_{ab}  = \nabla^{(h)}_a W_b + \nabla^{(h)}_b W_a -\frac{2}{n}(\div_h W) h_{ab}.
\end{equation}

We seek a solution $(h^*,K^*)$ of the constraint equations
where $[h^*]=[h]$ and where the conformal class velocity
of the initial data 
equals $[(h,U+\ck_h W)]$ for some vector field $W$ to be determined.
Recall $K^*$ encodes the
normal derivative of any extension of $h^*$ off of the 
initial slice via
\begin{equation}\label{eq:K-def}
2K^* = \frac{1}{N^*}(\partial_t h^*-\Lie_X h^*)
\end{equation}
where $N^* = dV_{h^*}/\alpha$ is the lapse determined by
$h^*$ and $\alpha$. This is essentially the Legendre transformation
at the level of the metric, and its trace-free part encodes
the Legendre transformation at the level of conformal classes.
Let $U^*$ be the trace-free part of $\partial_t h^*$
and let $\xi^*$  be the trace-free part of 
$K^*$. The conformal class velocity of the initial data is 
$[(h^*,U^*)]$ and equation \eqref{eq:K-def} implies
\[
2\xi^* = \frac{1}{N^*}(U^*-\ck_{h^*} W).
\]
By ansatz, $[h^*]=[h]$ and $[(h^*,U^*)]=[(h,U+\ck_h W)]$
and therefore $h^* = \phi^{q-2} h$ and $U^* = \phi^{q-2} (U + \ck_h W)$
for some conformal factor $\phi$. 
Using the relations $N^*=\phi^q N$
and $\ck_{h^*} = \phi^{q-2}\ck_h$ we find
\begin{equation}\label{eq:xi-def}
\xi^* = \frac{1}{2N^*}(U^* -\ck_{h^*} X)
= \phi^{-2} \frac{1}{2N} (U -\ck_{h} X + \ck_h W).
\end{equation}
Inserting $h^*=\phi^{q-2} h$ and $K^* = \xi^* +(\tau/n)h^*$
into the constraint 
equations \eqref{eq:ham-constraint}--\eqref{eq:mom-constraint}, and using equation
\eqref{eq:xi-def} to rewrite all quantities in terms of 
$(h,U,N,X,\tau)$ along with the unknowns $\phi$ and $W$, yields
\begin{align}
-2\kappa q \Delta_h \phi + R_h \phi -\frac{1}{4N^2}|U+\ck_h Z|_h^2\phi^{-q-1} + \kappa \tau^2 \phi^{q-1} &= 16\pi \phi^{q-1}{\hamE}^*
\label{eq:Lich-Lag}\\
\div_h\left(\frac{1}{2N}\ck_h Z\right) &= -\div_h\left(\frac{1}{2N}U\right) + 
\phi^q \kappa d\tau - 8\pi \phi^q{\momJ}^*,\label{eq:LCBY-mom-Lag}
\end{align}
where we have introduced $Z=W-X$ as a new unknown, and
where the matter terms ${\hamE}^*$ and ${\momJ}^*$ are functions of the matter seed variables in addition to, conceivably,
$(h,N,X,\phi,U,Z,\tau)$ and their derivatives.
Equation \eqref{eq:LCBY-mom-Lag} corresponds to equation
\eqref{eq:EMCSF-constraint-Lag} with $Z$ in equation \eqref{eq:LCBY-mom-Lag} to $\Theta$ in equation \eqref{eq:EMCSF-constraint-Lag}.
This presentation of the conformal method is often referred to as the
conformal thin-sandwich method \cite{york_conformal_1999}, but
it is identical, in a way made precise in \cite{maxwell_conformal_2014}, to the historical conformal method developed in the 1970s.

As an alternative to specifying a conformal class velocity $[(h,U)]$, 
one can provide a vacuum conformal class momentum (modulo spatial
diffeomorphisms) and request that the solution of the constraints
has the same conformal class velocity as would be determined by
the vacuum momentum.  This is 
the Hamiltonian formulation of the conformal method from
\cite{pfeiffer_extrinsic_2003} and is an analog
of the technique that leads to equation \eqref{eq:EMCSF-constraint-Ham}
for the EMCSF system. Vacuum conformal class momentum 
can be expressed, with respect to the conformal class 
representative $h$,
by
a symmetric, trace-free, divergence-free (0,2) tensor $\sigma$;
see, e.g., \cite{maxwell_conformal_2014}.
The variable $\sigma$ is analogous to $\tilde E$ in the
technique leading to equation \eqref{eq:EMCSF-constraint-Ham}
and  we seek a second-fundamental form 
\begin{align*}
K^* &= \phi^{-2}\sigma + \frac{1}{2 N^*} \ck_{h^*} W + \frac{\tau}{n}h^*\\
&= \phi^{-2}(\sigma + \frac{1}{2N}\ck_{h}W ) + \phi^{q-2}\frac{\tau}{n}h
\end{align*}
where the vector field $W$ is to be determined. 
The LCBY equations become
\begin{align}
-2\kappa q \Delta_h \phi + R_h \phi -\left|\sigma+\frac{1}{2N}\ck_h W\right|_h^2\phi^{-q-1} + \kappa \tau^2 \phi^{q-1} &= 16\pi \phi^{q-1}\hamE^*
\label{eq:Lich-Ham}
\\
\div_h\left(\frac{1}{2N}\ck_h W\right) &=  
\phi^q \kappa d\tau - 8\pi \phi^q\momJ^*.
\label{eq:LCBY-mom-Ham}
\end{align}
and in this (Hamiltonian) formulation, the momentum constraint 
\eqref{eq:LCBY-mom-Ham}
is an analog of equation \eqref{eq:EMCSF-constraint-Ham}.  

In the same way that $A_\vdash$ does not appear
in either of equations \eqref{eq:EMCSF-constraint-Lag}
or \eqref{eq:EMCSF-constraint-Ham}, the shift $X$
is no longer present in either formulation
\eqref{eq:Lich-Lag}-\eqref{eq:LCBY-mom-Lag} or
\eqref{eq:Lich-Ham}-\eqref{eq:LCBY-mom-Ham}.
It can be freely traded with $W$ in the Lagrangian
formulation and only the difference
$Z=X-W$ matters.  In the Hamiltonian formulation the
shift plays no role in the first place.
 So although 
we have formally included it as seed data to help describe the 
conformal method, it is not truly a parameter and can
be ignored.  By contrast, the role of the slice energy
density $\alpha$ is essential.  It appears in the LCBY equations implicity
via $N= dV_h/\alpha$ and changing $\alpha$ affects the solution
in the same way that changing $N$ in the EMCSF equations
\eqref{eq:EMCSF-constraint-Lag}
and \eqref{eq:EMCSF-constraint-Ham} affects the solution.
Ultimately, both formulations of the conformal method specify
a conformal class velocity (explicitly in the Lagrangian case
and implicitly in the Hamiltonian case), but the Legendre 
transformation linking the
conformal class momentum and then conformal class velocity
requires additional metric
information encoded in $\alpha$.  Without this additional
gauge information the problem posed
by the conformal method is not well specified.

\subsection{CMC Specialization}\label{secsec:CMC}

We are mostly interested in the constant
mean curvature (CMC) case where $\tau$
is constant and where the conformal
method has been maximally successful.  In this case
the LCBY momentum constraint \eqref{eq:LCBY-mom-Ham} in the Hamiltonian formulation 
becomes
\begin{equation}
\div_h\left(\frac{1}{2N}\ck_h W\right) =  
- 8\pi \phi^q  {\momJ}^*.
\end{equation}
Under the scaling law $\momJ^* = \phi^{-q}\momJ$
discussed at the start of Section \ref{sec:constraint-solving}
(and implied by Meta-Theorem \ref{mt:meta-theorem}), this
equation becomes
\begin{equation}\label{eq:LCBY-mom-Ham-2}
\div_h\left(\frac{1}{2N}\ck_h W\right) =  
- 8\pi {\momJ},
\end{equation}
which can potentially be solved for $W$ independent of $\phi$, in which
case analysis of the system reduces to the Lichnerowicz 
equation \eqref{eq:Lich-Ham} alone.  In the CMC Lagragian formulation,
again assuming the same scaling law for $\momJ$,
the momentum constraint is
\begin{equation}\label{eq:LCBY-mom-Lag-2}
\div_h\left(\frac{1}{2N}\ck_h Z\right) = -\div_h\left(\frac{1}{2N}U\right)  - 8\pi {\momJ}
\end{equation}
which is similarly independent of $\phi$.

There are two important differences between the 
vacuum and the non-vacuum initial data in the CMC conformal
method that we wish to highlight here.  In vacuum,
\begin{itemize}
  \item the presence of conformal Killing fields (nontrivial solutions of $\ck_h W =0$) plays no role, and
  \item the choice of slice energy density $\alpha$
  plays no role in the Hamiltonian formulation.
\end{itemize}
Both of these statements are false in the non-vacuum setting, a
fact that is easy to overlook for practitioners accustomed to
working with vacuum data sets.

Concerning conformal Killing fields, consider the vacuum CMC
Lagrangian LCBY momentum constraint
\begin{equation}\label{eq:lcby-mom-2}
\div_h\left(\frac{1}{2N}\ck_h Z\right) =  -\div_h\left(\frac{1}{2N}U\right)
\end{equation}
and for simplicity, consider the case of a compact manifold.
The left-hand side of equation \eqref{eq:lcby-mom-2}
is formally
self-adjoint with its kernel consisting of conformal Killing fields
of the metric $h$, should these exist. Integration by parts shows that 
the right-hand side is $L^2$ orthogonal
to any conformal Killing fields
and hence there exists solutions to equation \eqref{eq:lcby-mom-2}
whether or not conformal Killing fields exist for $h$. 
One readily verifies that multiple solutions to equation \eqref{eq:lcby-mom-2} exist if and only if $h$ admits
conformal Killing fields, since a pair of solutions
$Z'$ and $Z$ differ by a conformal Killing field. 
Since a solution $Z$ only appears in the LCBY Hamiltonian
constraint in the form $\ck_h Z = \ck_h Z'$, the potential multiplicity
of solutions is irrelevant.  On the other hand, for the full
equation \eqref{eq:LCBY-mom-Lag-2} we see that conformal Killing fields
pose a genuine obstruction to solution. The equation is
solvable if and only if $\momJ(h,B,\Pi_B)$
is $L^2$ orthogonal to any conformal Killing fields, and this poses
an additional constraint on the matter fields that needs 
to be satisfied in advance.

Concerning the role of the slice energy density $\alpha$ in
the vacuum CMC Hamiltonian setting (again, on a compact
manifold for simplicity), observe that the
solutions $W$ of equation \eqref{eq:LCBY-mom-Ham-2} 
with $\momJ=0$ are exactly conformal Killing fields.  
Since $\alpha$ appears in the Lichnerowicz equation 
\eqref{eq:Lich-Ham} only via the lapse $N$ in
the form $1/(2N) \ck_h W$, it follows that the solutions of this system,
if any, are independent of the choice of $\alpha$.   But
in the non-vacuum case the choice of $\alpha$ has 
an impact on the resulting solution: changing $\alpha$
changes $N$ in equation \eqref{eq:LCBY-mom-Ham-2} 
and, as can be seen directly in simple examples,
thereby changes solutions of the non-homogeneous equation.
This goes on to impact the term $1/(2N)\ck_h W$ in the Lichnerowicz
equation and ultimately the
resulting solution of the constraint equations.  The significance
of the dependence of the solution on the choice of $\alpha$
is not yet well understood, and deserves further study.

Notably, these differences between the vacuum and non-vacuum
CMC cases also appear in the non-CMC setting, even in vacuum.  
The slice energy density $\alpha$ 
appears throughout the coupled system 
in the form of $N$, and the presence of conformal Killing
fields interferes with finding solutions of the LCBY
vacuum momentum constraint
\[
\div_h\left(\frac{1}{2N}\ck_h W\right) = \phi^q \kappa d\tau
\]
because the right-hand side need not remain in the image
of $\div_h$ as the unknown $\phi$ varies.  Because of these
connections, it may
be that a better understanding of the CMC non-vacuum conformal
method leads to insights regarding 
an improved variation of the conformal
method for non-CMC solutions.

\section{Applications}
\label{sec:apps}
We now apply our construction to derive the right-hand sides 
of the LCBY equations for a selection of concrete matter models.
For definiteness, consider the Hamiltonian formulation
in the CMC setting discussed in Sections \ref{secsec:ConfMethod}
and \ref{secsec:CMC}. Thus the metric seed data 
can be taken to consist of a Riemannian metric $h$, a transverse-traceless tensor
$\sigma$, a constant mean curvature $\tau$ 
and a positive lapse $N$; the LCBY equations then read
\begin{align}
\label{eq:Lich-final}
-2\kappa q \Delta_h \phi + R_h \phi -\left|\sigma+\frac{1}{2N}\ck_h W\right|_h^2\phi^{-q-1} + \kappa \tau^2 \phi^{q-1} &= 16\pi \phi^{q-1}\hamE^*\\
\label{eq:LCBY-mom-final}
\div_h\left(\frac{1}{2N}\ck_h W\right) &=  
- 8\pi \phi^q  {\momJ}^*.
\end{align}
Our aim for each matter model 
is to write down the form of $\hamE^*$ and $\momJ^*$
as functions of $\phi$ along with other seed data
to derive the final equations.

In each case the broad program begins with the following steps:
\begin{enumerate}
  \item Identify the Lagrangian $\mathcal L$ for the matter model along
  with the spacetime field variable(s).
  \item For tensor-valued fields, decompose into spatial variables
  by the procedure of Section \ref{sec:decomp}.
  \item Determine the energy and momentum densities $\hamE$
and $\momJ$ for the model either by computing the
stress-energy tensor and using equations \eqref{eq:E-def-direct}--\eqref{eq:J-def-direct} or by taking derivatives of the Lagrangian with respect to the lapse and shift via equations \eqref{eq:def-J-first}--\eqref{eq:def-E-first}.
\item Letting $B$ denote the (potentially spatially decomposed) matter
fields, identify the conjugate momenta $\Pi_B$ by writing the Lagrangian
in terms of $B$ and its derivatives and computing
the derivative of the Lagrangian with respect to $\dot B$.
\item Identify constraints for the matter fields arising from the Euler-Lagrange equations
for the Lagrangian and determine if they are invariant under conformal changes; a failure of conformal invariance at this step would require 
some other procedure to construct initial data.
\end{enumerate}

At this point, perhaps after a computation, we can write
\begin{align}
\label{eq:E-from-B-h}
\hamE &= \hamE(B,\Pi_B,h,N,X)\\
\label{eq:J-from-B-h}
\momJ &= \momJ(B,\Pi_B,h,N,X).
\end{align}
The quantities $\hamE^*$ and $\momJ^*$
in equations \eqref{eq:Lich-final}--\eqref{eq:LCBY-mom-final}
are then
obtained by replacing $h$ with $h^*=\phi^{q-2}h$ and $N^*=\phi^q N$
in equations \eqref{eq:E-from-B-h}--\eqref{eq:J-from-B-h}:
\begin{align*}
\hamE^* &= \hamE(B,\Pi_B,\phi^{q-2}h,\phi^q N,X)\\
\momJ^* &= \momJ(B,\Pi_B,\phi^{q-2}h,\phi^q N,X),
\end{align*}
which completes the procedure. Because we have fixed
$B$ and $\Pi_B$, Theorems \ref{thm:form-of-J-manifold-edition}
and \ref{thm:form-of-J-tensor-edition} imply that
the momentum constraint with this form of $\momJ^*$
will have decoupled from the Hamiltonian constraint.

In practice, it is frequently cumbersome to write $\hamE$
and $\momJ$ in terms of $B$ and $\Pi_B$ directly, in part
because familiar notation for matter models use quantities 
that encode $B$ and $\Pi_B$ in terms of variables that mix in the metric.
For example, the electric field $E$ 
determines a momentum $-1/(2\pi)\ip<E,\cdot>_h\;dV_h$.  
This momentum can be held fixed when $h$ conformally changes so long
as $E$ conformally changes as well; $h^*=\phi^{q-2} h$ leads to
$E^*=\phi^{-2} E$.  In the case of fluids, there is a conformally invariant $n$-form $\omega$ determining particle number density,
and this is determined by a conformally transforming 
scalar $\mathbf n$ via $\omega= \mathbf n\,dV_h$; 
maintaining $\omega=\mathbf n^*\,dV_{h^*}$ requires $\mathbf n^* = \phi^{-q}\mathbf n$. In general we work with conformally
transforming variables $C$ and transformation rules chosen 
so that $(C,h,N)\mapsto (C^*,h^*,N^*)=(C^*(C,h,\phi,N), \phi^{q-2} h, \phi^q N)$
ensures $B$ and $\Pi_B$ remain fixed.  The right-hand sides
of the LCBY equations are then computed using
\begin{align}
\label{eq:E-conf-variable}
\hamE^* &= \hamE(C^*(C,h,\phi,N),\phi^{q-2}h,\phi^q N,X)\\
\label{eq:J-conf-variable}
\momJ^* &= \momJ(C^*(C,h,\phi,N),\phi^{q-2}h,\phi^q N,X).
\end{align}

\subsection{Perfect Fluids}
\label{sec:fluids}

Perfect fluids admit a Lagrangian description in which the field
can be thought of as taking values in a separate manifold, the
so-called material manifold.  Hence Theorem 
\ref{thm:form-of-J-manifold-edition} applies to these fields.  Because
the Lagrangian formulation for fluids is perhaps less familiar
than the stress-energy tensor it determines, we begin
with a brief summary of it.  For more details
the reader is referred to \cite{brown_action_1993} and
the approachable notes \cite{Volker}.

\subsubsection{Lagrangian Formulation}
The $n$-dimensional \textbf{material manifold} $\overline\Sigma$ is a set of
labels for infinitesimal clusters of fluid particles. We make
the important simplifying hypothesis
\footnote{For example, we may be modeling the interior of a fluid body. A Lagrangian description of a fluid only partly filling space requires
a formulation of field theory more general than the one employed
in Sections \ref{sec:manifold-valued} and \ref{sec:tensor-valued}.} that the fluid fills all of spacetime $M=\Sigma\times I$, in
which case we can take $\overline\Sigma=\Sigma$.
The number density of particles on $\overline \Sigma$
is determined by a fixed positive $n$-form $\overline{\omega}$.
Together, $(\overline \Sigma,\overline\omega)$ can be thought of
as some instantaneous reference configuration of the fluid.

A \textbf{fluid configuration} on $M$ is determined by a submersion
$\Psi:M\to \overline \Sigma$ and the world line of
a particle cluster $q\in\overline\Sigma$ is $\Psi^{-1}(q)$.  
The map $\Psi$ here
corresponds to the map $\secA$ of Section \ref{sec:manifold-valued},
with $\overline \Sigma$ playing the role of $F$ from that section.
The flux of particles through $n$-dimensional subspaces
of $M$ is determined by the (non-vanishing) $n$-form
\[
\omega = \Psi^* \overline \omega,
\]
and the fact that
$d(\Psi^* \omega)= \Psi^*(d\omega)=0$ reflects conservation
of particle number for the fluid. Note that this 
description of the fluid configuration is completely independent 
of any spacetime metric.

If we have a spacetime metric $g$ on hand, it is convenient to express
other tensorial quantities derived from $\omega$ in terms of
$g$.  First, the \textbf{particle flux} $Q$ is the unique
vector satisfying
\begin{equation}\label{eq:Qdef}
Q\interior dV_g = \omega.
\end{equation}
The particle flux is tangent to the worldlines of the fluid particles
as can be seen by observing that $\Psi_* Q = 0$.  
Indeed:
\[
\Psi^*( (\Psi_* Q) \interior \overline{\omega}) = Q\interior\omega = Q\interior Q\interior dV_g = 0,
\]
and the conclusion that $\Psi_* Q=0$ then follows from the observations 
that $\Psi$ is a submersion and that
$\overline\omega$ is a non-vanishing top-level form.
We assume that $Q$ is timelike
and future pointing, which is an open compatibility condition
between $\Psi$ and $g$.  Note that
\[
(\div_g Q) dV_g = d(Q\interior dV_g) = d\omega = 0
\]
and hence $\div_g Q= 0$; this is a metric-biased reflection
of the metric-free fact that particle number is conserved.
The \textbf{rest particle density} is the
scalar $\mathbf n_0>0$ defined by
\begin{equation}\label{eq:n0def}
\mathbf{n}_0^2 = - g(Q,Q)
\end{equation}
and the unit vector $U$ defined by
\[
Q = \mathbf{n}_0 U
\]
is the \textbf{fluid spacetime velocity}.
  
The Lagrangian for an isentropic perfect fluid 
is determined by a function 
$\rho(\mathbf{n}_0)$ that expresses mass-energy density in
terms of rest number density:
\begin{equation}\label{eq:fluid-Lag}
\mathcal L_{\mathrm{Fluid}}(\Psi,\partial\Psi,g) = -2\rho(\mathbf{n}_0(\Psi,\partial\Psi,g)) dV_g.
\end{equation}
Note that $\mathbf{n}_0$ is implicitly 
a function of both the metric (equations \eqref{eq:Qdef} and \eqref{eq:n0def})
as well as the field configuration $\Psi$ and its derivatives
($\mathbf n_0$ depends on $\omega = \Psi^*\overline \omega$).

The stress-energy tensor is obtained by varying the Lagrangian
with respect to the (inverse) metric, and to do this we
need to compute the variation of $\mathbf n_0$ with respect to 
the metric. This computation is simplified by fixing an arbitrary non-vanishing $n+1$-form $\Omega$
on spacetime so that $\Psi^*\overline{\omega} = Y\interior \Omega$
for some vector field $Y$ that does not depend on the metric. 
Observing $Q = Y(\Omega/dV_g)$
we can write
\[
\mathbf{n}^2_0 = -g(Y,Y) \left(\frac{\Omega}{dV_g}\right)^2.
\]
This expression for $\mathbf{n}^2_0$ has the advantage
that the dependence on $g$ is now explicit; $Y$ depends on $\Psi$
and our arbitrary choice of $\Omega$ but is independent of the metric.
Thus
\begin{align*}
\left(\frac{\partial \mathbf{n}_0^2}{\partial g^{-1}}\right)_{ab} &= 
Y_a Y_b \left(\frac{\Omega}{dV_g}\right)^2 -g(Y,Y)
\left(\frac{\Omega}{dV_g}\right)^2 g_{ab}\\
&= Q_a Q_b -g(Q,Q)g_{ab}\\
&= \mathbf{n}_0^2 \left(U_a U_b + g_{ab}\right)
\end{align*}
and as a consequence
\begin{align}
T_{ab}\,dV_g = -\left( \frac{\partial \mathcal L_{\mathrm{Fluid}}}{\partial g^{-1}}\right)_{ab} &=
(-\rho(\mathbf n_0) g_{ab}dV_g + \rho'(\mathbf n_0)\mathbf{n}_0[ U_a U_b + g_{ab}]) \,dV_g\nonumber\\
&= \left[\rho(\mathbf n_0) U_a U_b + (\rho'(\mathbf n_0)\mathbf{n}_0-\rho(\mathbf n_0))[ U_aU_b+g_{ab}]\right]\,dV_g.
\label{eq:fluid-T}
\end{align}
We recognize this as the stress energy tensor of a 
perfect fluid
with energy density $\rho(\mathbf n_0)$ and pressure $p$ determined by
the rest number density:
\begin{equation}\label{eq:pressure-def}
p(\mathbf n_0) = \rho'(\mathbf n_0) \mathbf{n_0}-\rho(\mathbf n_0).
\end{equation}
Noting that we are working with an isentropic fluid,
the equation of state \eqref{eq:pressure-def} is consistent with
the first law of thermodynamics; compare with 
\cite{misner_gravitation_2017} equations (22.6) and (22.7a).
So long as $\rho$ is monotone in $\mathbf n_0$,
equation \eqref{eq:pressure-def} implicitly gives
a more familiar equation of state $p = p(\rho)$.  For example,
linear equations of state $p=(\gamma-1)\rho$ used in cosmology
arise from
\[
\rho(\mathbf n_0) = C \mathbf n_0^\gamma
\]
for some constant $C$. The special case of a pressureless fluid (dust) 
corresponds to $\gamma = 1$ and a stiff fluid with the
speed of sound equal to the speed of light arises at
the upper threshold value $\gamma=2$; in Sections \ref{secsec:dust}
and \ref{secsec:hfluid} we derive the right-hand sides of the LCBY
equations for these two cases.

Before writing down the energy and momentum densities implied by
equation \eqref{eq:fluid-T} it is helpful to decompose
the particle flux:
\begin{equation}\label{eq:Q-def}
Q = \mathbf{n}( \nu + V)
\end{equation}
where $\nu$ is the unit normal to slices of constant $t$,
$V$ is the spatial velocity of the particle cluster measured
by an observer traveling orthogonal to the slices of constant $t$,  and $\mathbf{n}$ is the particle density seen by the same observer. 
Then
\begin{equation}
\ip<U,\nu>_g = \frac{1}{\mathbf n_0} \ip<Q,\nu>_g= 
-\frac{\mathbf n}{\mathbf n_0}
\end{equation}
and consequently
\begin{equation}
\hamE = T(\nu,\nu) = \rho(\mathbf n_0) + \left( \left(\frac{\mathbf n}{\mathbf n_0}\right)^2 -1 \right)\mathbf n_0 \rho'(\mathbf n_0)\label{eq:fluid-E-v1}.
\end{equation}
Since $j_*\ip<\nu,\cdot>_g=0$ and $j_*\ip<V,\cdot>_g=\ip<V,\cdot>_h$
it follows that
\begin{equation}
\momJ = -j_* T(\nu,\cdot) = \rho'(\mathbf n_0)\left(\frac{\mathbf n^2}{\mathbf n_0}\right)\ip<V,\cdot >_h.
\label{eq:fluid-J-v1}
\end{equation}
Equations \eqref{eq:n0def} and \eqref{eq:Q-def} imply
the fundamental relation determining rest particle density
from observed particle density and velocity:
\begin{equation}\label{eq:n0-vs-n}
\mathbf n_0^2 = \mathbf n^2(1-|V|_h^2).
\end{equation}
Hence we have an alternative form of the energy density:
\begin{equation}
\hamE =  \rho(\mathbf n_0) + \frac{\mathbf n^2}{\mathbf n_0}\rho'(\mathbf n_0)|V|_h^2.
\label{eq:fluid-E-v2}
\end{equation}

\subsubsection{Conjugate Momentum for Perfect Fluids}

Because our approach to constructing initial data involves
fixing the field value $\Psi$ and its conjugate momentum
$\Pi_\Psi$,
we need to identify this latter quantity. Setting $L_\Sigma=\partial_t\interior \mathcal L_{\mathrm{Fluid}}$,
 by definition
\begin{equation}\label{eq:Pi-fluid-part1}
\Pi_\Psi = \frac{\partial  L_\Sigma}{\partial \dot \Psi}
= -2\rho'(\mathbf n_0) \frac{\partial \mathbf n_0}{\partial \dot \Psi}\; N dV_h
\end{equation}
Derivatives of $\Psi$ appear in $\mathbf n_0$ 
because it depends on
$\omega=\Psi^*\overline \omega$ and we need to isolate the
dependence on $\dot\Psi$ while keeping $\partial_x \Psi$ fixed.
Since $\mathbf n_0$ is determined from $\mathbf n$ and $V$ along with
the metric (equation \eqref{eq:n0-vs-n}),
it suffices to determine how $\mathbf n$ and $V$ depend on 
$\dot \Psi$.
\begin{lemma}\label{lem:n-V}
The following relations hold:
\begin{align}
\label{eq:n-on-Psi_x}
 \mathbf n\; dV_h &= j_* \Psi^* \overline \omega \\
\label{eq:V-on-Psi_dot}
\dot\Psi &= \Psi_* (-NV +X).
\end{align}
\end{lemma}
\begin{proof}
To establish \eqref{eq:n-on-Psi_x} we compute
\[
\Psi^* \overline\omega = Q\interior dV_g = \mathbf n(\nu+V)\interior dV_g = 
(\mathbf n/N)(\partial_t + Z)\interior dV_g
\]
where $Z$ is an unimportant spatial vector. 
Lemma \ref{lem:t-preserving-diffeo}
part \ref{part:kills-S-interior-omega} 
implies $j_* Z\interior dV_g=0$ and hence
\[
j_*\Psi^* \overline\omega = \frac{\mathbf n}{N} j_* \partial_t\interior dV_g = \mathbf n\;dV_h.
\]

Turning to equation \eqref{eq:V-on-Psi_dot} we recall $\Psi_* Q=0$.
Since $Q=\mathbf n(\nu+V)$ it follows that
\[
\Psi_*\nu = - \Psi_*V
\]
and equation \eqref{eq:V-on-Psi_dot} results from an easy
computation using the identities
$\nu = (\partial_t-X)/N$ and $\dot\Psi=\Psi_*\partial_t$.
\end{proof}

Lemma \ref{lem:n-V} has two implications for the computation
of $\Pi_\Psi$.  First, $\mathbf n$ is computable from $\partial_x\Psi$
and $h$ alone and is independent of $\dot\Psi$.  So equation
\eqref{eq:n0-vs-n} implies
\begin{equation}\label{eq:Pi-fluid-part2}
\frac{\partial \mathbf n_0^2}{\partial \dot\Psi} = -2\mathbf n^2 \ip<V,\partial V/\partial \dot \Psi>_h.
\end{equation}
Second, equation \eqref{eq:V-on-Psi_dot} permits computation of $\partial V/\partial \dot \Psi$.  To do so, define 
\[
\psi(t) = \Psi|_{\Sigma_t}.
\]
Since $\Psi$ is, by hypothesis, a submersion and
since the (assumed timelike) vector $Q$ lies in the kernel of $\Psi_*$, a dimension argument shows that for each $t$, $\psi(t)$ 
is a local diffeomorphism.
Thus $\psi(t)^*$ is well-defined on tangent vectors and equation
\eqref{eq:V-on-Psi_dot} implies
\[
\psi(t)^*\dot\Psi = (-NV + X)|_{\Sigma_t}.
\]
Tangent vectors on $\Sigma_t$ can be identified with tangent
vectors on $M$ and we conclude
\begin{equation}\label{eq:V-vs-dotPsi}
V = -\frac{\psi^* \dot\Psi -X}{N}
\end{equation}
where $\psi^* \dot\Psi$ is defined to be the 
unique spatial vector $Y$ on $M$ with 
$\Psi_*Y = \dot\Psi\in T\overline\Sigma$.
Noting that $\psi^*$ depends on $\partial_x\Psi$ but not on $\dot\Psi$,
we conclude from equations \eqref{eq:Pi-fluid-part1}, 
\eqref{eq:Pi-fluid-part2} and \eqref{eq:V-vs-dotPsi} that
\begin{equation}\label{eq:pi-fluid}
\Pi_{\Psi} = -2\rho'(\mathbf n_0)\frac{\mathbf n^2}{\mathbf n_0} \ip<V,\psi^* \cdot>_h\;dV_h.
\end{equation}
This object is a little unwieldy; it is a $T^*\overline \Sigma$-valued
$n$-form on $M$.  To mitigate this, we 
introduce the related covector field
\begin{equation}\label{eq:P-psi-def}
P_\Psi = \rho'(\mathbf n_0)\frac{\mathbf n}{\mathbf n_0} \ip<V,\cdot>_h
\end{equation}
so that 
\begin{equation}\label{eq:Pi_Psi_vs_P_psi}
\Pi_\Psi = -2(P_\Psi\circ \psi^*)\mathbf n\,dV_h,
\end{equation}
and we interpret $P_\Psi$ as
the fluid momentum density per unit particle.
Equation \eqref{eq:n-on-Psi_x} implies $\mathbf n\,dV_h$ is
metric-independent
and hence $P_\Psi$ is computable from $\Psi$ and $\Pi_\Psi$
alone without using a spacetime metric.

We can now rewrite the energy and momentum densities
from equations \eqref{eq:fluid-J-v1} and \eqref{eq:fluid-E-v2}
 as functions
of $\Psi$, $\Pi_\Psi$ and $h$ as follows:
\begin{align}
\label{eq:fluid-E-v3}
\hamE &= \rho(\mathbf n_0) + \frac{\mathbf n_0}{\rho'(\mathbf n_0)}|P_\Psi|_h^2\\
\label{eq:fluid-J-v3}
\momJ &= \mathbf n\, P_\Psi
\end{align}
where $\mathbf n$ is determined from $\Psi$ and $h$ via
$\mathbf n\;dV_h = j_*\Psi^* \overline \omega$, 
where $P_\Psi$ is computed from $\Psi$ and $\Pi_\Psi$ (but not $h$) as 
discussed above, and where $\mathbf n_0$ is implicitly determined
from $\mathbf n$, $P_\Psi$ and $h$ by
equations \eqref{eq:n0-vs-n} and \eqref{eq:P-psi-def}:
\begin{equation}\label{eq:n0-from-n-final}
\left(\frac{\mathbf n}{\mathbf n_0}\right)^2 = 1 + \frac{1}{(\rho'(\mathbf n_0))^2}|P_\Psi|_h^2.
\end{equation}
Note that taking equation \eqref{eq:Pi_Psi_vs_P_psi} into account,
the momentum constraint is exactly of the form of equation 
\eqref{eq:J-to-Pi-B-manifold-edition} from Theorem \ref{thm:form-of-J-manifold-edition}.

\subsubsection{Application of the Conformal Method}
\label{secsec:conf-fluid}

We work on a single slice $\Sigma_{t_0}$ which we 
identify canonically with $\Sigma$. Spatial tensors
defined along on $\Sigma_{t_0}$ in $M$ in the
discussion above
can then be unambiguously identified with tensors on $\Sigma$, and
we do so without comment.  Assuming that 
the initial fluid configuration $\psi(t_0)=\Psi|_{\Sigma_{t_0}}$ 
is a diffeomorphism, without loss of generality 
we can take $\overline \Sigma=\Sigma$ and $\psi(t_0)=\Id$.
Hence the reference
number density $\overline\omega$ is simply the
initial number density which, in a mild abuse of notation,
we write as $\omega$.\footnote{Strictly speaking, 
in terms of our notation used up to this point
the 
initial number density is the spatial tensor $j_*\omega$}
Having chosen $\omega$, the momentum $\Pi_\Psi$ is determined
by choosing a covector field $P_\Psi$ according to
equations \eqref{eq:n-on-Psi_x} and \eqref{eq:Pi_Psi_vs_P_psi}
with $\psi=\Id$:
\[
\Pi_\Psi = -2P_\Psi\otimes\omega.
\]
That is, the metric-independent matter field and 
the conjugate momentum
are determined from $(\omega,P_\Psi)$, which are taken
as the fixed non-metric seed data for the conformal method.  

Since $\hamE$ and $\momJ$ in equations
\eqref{eq:fluid-E-v3}--\eqref{eq:fluid-J-v3}
are written in terms of $(\mathbf n,\mathbf n_0,P_\Psi,h)$
and do not explicitly involve our fixed variable $\omega$,
we follow the procedure discussed at the start of Section 
\ref{sec:apps} and introduce conformal transformation
rules
\[
(\mathbf n,\mathbf n_0,P_\Psi,h)\mapsto
\left(\mathbf n^*,\mathbf n_0^*,P_\Psi,h^*=\phi^{q-2}h\right)
\]
that are equivalent to the 
rule $(\omega,P_\Psi,h)\mapsto(\omega,P_\Psi,h^*)$.
In place of $\omega$ for the seed data we
specify the function $\mathbf n$ 
and set $\omega =\mathbf n\, dV_h$.  Keeping $\omega$ fixed
we conclude from $\omega=\mathbf n^*dV_{h^*}$ that
$\mathbf n^*=\phi^{-q}\mathbf n$.  At this point, $\mathbf n_0^*(\phi)$ is determined by substituting starred variables into the implicit relation \eqref{eq:n0-from-n-final} and we find
\begin{equation}\label{eq:n_0_star}
\left(\frac{\mathbf n}{\mathbf n_0^*(\phi)}\right)^2\phi^{-2q} = 1 + \frac{1}{(\rho'(\mathbf n_0^*(\phi)))^2}\phi^{2-q}|P_\Psi|_h^2.
\end{equation}

The right-hand side of the momentum constraint is
\begin{equation}
-8\pi\phi^q \momJ^* = -8\pi\phi^q \mathbf n^* P_\Psi = -8\pi \mathbf n \,P_\Psi
\end{equation}
and consequently the momentum constraint reads
\begin{equation}\label{eq:fluid-mom-const}
\div_h\left(\frac{1}{2N}\ck_h W\right) =
-8\pi  \mathbf n\,  P_\Psi.
\end{equation}
As expected, equation \eqref{eq:fluid-mom-const}
does not involve $\phi$ and  has decoupled from the Hamiltonian
constraint. Note that equation \eqref{eq:fluid-mom-const} holds regardless of the fluid's equation of state.  

By contrast, the form of the Hamiltonian constraint is
sensitive to the choice of $\rho(\mathbf n_0)$, and we compute
\[
\hamE^* = \rho(\mathbf n_0^*(\phi)) + 
\frac{\mathbf n_0^*(\phi)}{\rho'(\mathbf n_0^*(\phi))}\phi^{2-q}|P_\Psi|_h^2
\]
with $\mathbf n_0^*(\phi)$ determined implicitly 
by equation \eqref{eq:n_0_star}.  This implicit relationship
complicates the analysis of the conformal method, and a 
full treatment is beyond the scope of this paper.  
We content ourselves with two examples that illustrate 
some of the possibilities, dust and stiff fluids. 

\subsubsection{Dust}\label{secsec:dust}
Dust arises from an equation of state
$\rho(\mathbf n_0) = m \mathbf n_0$ where $m$ is the particle rest 
mass, and a computation shows
\begin{equation}\label{eq:Estar-cosmo}
\hamE = m \frac{\mathbf n^2}{\mathbf n_0}.
\end{equation}
Moreover, because $\rho'(\mathbf n_0)\equiv m$, we can solve equation \eqref{eq:n0-from-n-final} explicitly for $\mathbf n/\mathbf n_0$ to obtain
\[
\hamE = \mathbf n  \sqrt{m^2+|P_\Psi|_h^2}.
\]
Therefore
\[
\hamE^* = \mathbf n^*  \sqrt{m^2+|P_\Psi|_{h^*}^2}
= \mathbf n \phi^{-q} \sqrt{m^2+\phi^{2-q}|P_\Psi|_{h}^2}
\]
and the Lichnerowicz equation \eqref{eq:Lich-final} becomes
\begin{equation}\label{eq:Lich-dust}
-2\kappa q \Delta_h \phi + R_h \phi -\left|\sigma+\frac{1}{2N}\ck_h W\right|_h^2\phi^{-q-1} + \kappa \tau^2 \phi^{q-1} = 
16\pi \mathbf n \sqrt{m^2\phi^{-2} + \phi^{-q}|P_\Psi|_h^2}.
\end{equation}
Because the right-hand side of this equation is monotone decreasing in $\phi$, barrier techniques such as those employed in \cite{Isenberg:1995bi}
can be used to establish existence and uniqueness of $\phi$
in many contexts.

Equation \eqref{eq:Lich-dust} that we have derived here 
is manifestly different from the Lichnerowicz equation 
obtained for fluids 
in the earlier works, e.g. \cite{dain_initial_2002} or \cite{isenberg_gluing_2005}.  
These 
papers treat $\hamE$ and $\momJ$ as conformally
transforming quantities with transformation rules
independent of the fluid equation of state.  For example,
\cite{isenberg_gluing_2005} uses the rule
\begin{align*}
\mathcal{E}^* = \phi^{1-\frac32 q}\hamE\\
\mathcal{J}^* = \phi^{-q}\momJ.
\end{align*}
The transformation law for $\momJ$ is chosen simply for pragmatic
reasons to ensure that the momentum constraint decouples, and the
transformation law for $\hamE$ is chosen to ensure
\[
\frac{|J^*|_{h^*}^2}{(\hamE^*)^2}=\frac{|J|_{h}^2}{\hamE^2}
\]
and hence the dominant energy condition (DEC) holds for 
$(\hamE^*,\momJ^*)$ if it does for
$(\hamE,\momJ)$.  Alternatively,
\cite{dain_initial_2002} uses the same transformation 
law for $\momJ$
but keeps $\hamE$ conformally invariant, motivated by
by important considerations arising at the boundary of a 
compact fluid body not filling spacetime.
In both works, once a physical $\hamE^*$
and $\momJ^*$ have been determined by solving the LCBY equations,
physical fluid variables (energy $\rho$ and velocity $V$) are extracted by inverting the relationship that defines $(\hamE^*,\momJ^*)$
from $(\rho,V)$ via the fluid equation of state; this operation is shown
to be possible under natural physical assumptions about the 
equation of state, assuming $(\hamE^*,\momJ^*)$ satisfies
the DEC. For example, the deconstruction procedure in \cite{isenberg_gluing_2005}
implies the conformal transformation rules
\begin{align}
\rho^* &= \phi^{1-3q/2} \rho\label{eq:old-rho-trans}\\
V^* &= \phi^{q/2-1}V\label{eq:old-Z-trans}.
\end{align}

Although the approach we have described here is somewhat more involved
than these earlier approaches, it enjoys a clearer connection between what is prescribed from matter in the seed data versus what appears in
the finally constructed initial data.  We are specifying a 
distribution 
of particles $\omega$ along with a conjugate momentum that are conformally invariant.  In particular, the particle count $\int_{\Omega} \omega$ is preserved at the end of the construction in any region $\Omega\subseteq\Sigma$, and this better connects the Lichnerowicz equation to the Poisson equation of Newtonian gravity.
By contrast, when $\rho$ and $V$ conformally transform 
with the rules \eqref{eq:old-rho-trans}--\eqref{eq:old-Z-trans}
the final particle count depends on the determined conformal factor; e.g.,
if $V=0$ a computation shows $\omega^* = \phi^{1-q/2}\omega$.

\subsubsection{Stiff Fluids}\label{secsec:hfluid}
In the earlier constructions of \cite{dain_initial_2002} and \cite{isenberg_gluing_2005}
the fluid's equation of state plays a compartmentalized role: it is used initially to determine
$\hamE$ and $\momJ$, and it is used after solving the LCBY
equations 
to deconstruct $\hamE^*$ and $\momJ^*$ into fluid variables.  But it plays no role in the solvability of the LCBY system.  So,
from a certain perspective, with such a strategy
all fluids are equivalent, assuming they satisfy
the hypotheses needed to deconstruct $(\hamE^*,\momJ^*)$.
In our procedure, the equation of state is an important component 
of the construction,
and in this section we show how the LCBY
equations for stiff fluids are closer in nature to those for scalar 
fields than they are to those for dust.

The equation of state of a stiff fluid 
is $\rho(\mathbf n_0) = \mu\mathbf n_0^2$
for some constant $\mu$ and equation \eqref{eq:fluid-E-v1} implies
\[
\hamE = \mu (2 \mathbf n^2-\mathbf n_0^2).
\]
Equation \eqref{eq:n0-from-n-final} for $\mathbf n_0$
can be rewritten
\[
\mathbf n^2 = \mathbf n_0^2 + \frac{1}{\mu^2}|P_\Psi|_{h}^2.
\]
and as a consequence
\begin{align*}
\hamE &= \mu \mathbf n^2 + \frac{1}{\mu}|P_\Psi|_{h}^2.
\end{align*}
Similarly
\begin{align*}
\hamE^* &= \mu \mathbf (\mathbf n^*)^2 + \frac{1}{\mu}|P_\Psi|_{h^*}^2\\
&=  \mu \phi^{-2q}\mathbf n^2 + \frac{1}{\mu} \phi^{2-q}|P_\Psi|_{h}^2.
\end{align*}
and the Hamiltonian constraint becomes
\begin{equation}\label{eq:lich-hard-fluid}
-2\kappa q \Delta_h \phi + R_h \phi -\left|\sigma+\frac{1}{2N}\ck_h W\right|_h^2\phi^{-q-1} + \kappa \tau^2 \phi^{q-1} = 
16\pi \left[  \phi^{-q-1}\mu\mathbf n +  \frac{1}{\mu} \phi |P_\Psi|_{h}^2\right].
\end{equation}
The powers of $\phi$ appearing in the terms on the right-hand side of equation \eqref{eq:lich-hard-fluid}
strikingly correspond to those on the left-hand side.  
The $\phi^{-q-1}$ term is
monotone decreasing in $\phi$ and therefore poses no difficulty for solvability. By contrast, the $\phi^1$ term is monotone increasing and needs to be accounted for by techniques such those found in 
\cite{hebey_variational_2008} for working with scalar field matter sources. 
Indeed, equation \eqref{eq:lich-hard-fluid} has
a structure identical to that 
of the Lichnerowicz equation for a scalar
field, equation \eqref{eq:scalar-field-Ham}.  The correspondence
found here between stiff fluids and scalar fields echoes
well-known connections between the two models 
in the associated evolution problem 
(\cite{christodoulou_formation_2007}, e.g.) and
is absent in previous initial data constructions.

\subsection{Proca}
\label{secsec:proca}

The Proca field describes a massive connection on a $U(1)$ bundle,
generalizing the connection of Maxwell's equations.  
The field is a 1-form $\mathcal A$ on spacetime and
the Lagrangian  is
the sum $\mathcal L_{\mathrm{Proca}}=\mathcal L_{\mathrm{EM}}+\mathcal L_{m}$ of the standard electromagnetism Lagrangian
\[
\mathcal L_{\mathrm{EM}} = -\frac{1}{8\pi} |d\mathcal A|_g^2\;dV_g
\]
and a term involving the constant mass $m$ of the field
\[
\mathcal L_m = -\frac{m}{4\pi}|\mathcal A|_g^2\;dV_g.
\]
As in Section \ref{sec:EMCSF} we use the following the notation:
\begin{align*}
E &= -j_* (\nu \interior d \mathcal A)\\
A &= j_* \mathcal A\\
A_0 &= \partial_t \interior \mathcal A\\
A_\perp &= \nu\interior \mathcal A\\
A_\vdash &= (\partial_t-X)\interior \mathcal A = A_0 - X\interior A\\
\dot A &= \partial_t A.
\end{align*}
These are all spatial tensors, and we can identify them with 
tensors on an initial slice $\Sigma_{t_0}$, which we
additionally identify with $\Sigma$.
The computation of Section \ref{sec:EMCSF} shows that on $\Sigma$
\begin{align}\label{eq:E-Proca-def}
E &= -\frac{1}{N}\left(\dot A - d A_\vdash - \Lie_X A\right)\\
\label{eq:E-Proca-def-A0}
&=-\frac1N\left(\dot A -  d A_0 - X\interior d A\right)
\end{align}
where $d$ is the exterior derivative intrinsic to $\Sigma$.
Using these variables, the slice Lagrangian $L_\Sigma = \partial_t\interior \mathcal L_{\mathrm{Proca}}$ can be written
\begin{equation}\label{eq:proca-slice-lag}
L_{\Sigma} = \left[ \frac{1}{4\pi}|E|_h^2 -\frac{1}{8\pi}| dA|_h^2 + \frac{m^2}{4\pi} \left[ A_{\perp}^2 -|A|_h^2\right]\right]NdV_h.
\end{equation}
In light of equation \eqref{eq:E-Proca-def} 
and the identity $A_\perp = A_\vdash/N$ one 
can think of $L_\Sigma$ as a function of metric variables $(h,N,X)$
and field variables $(A,A_\vdash,\dot A, \dot A_{\vdash})$.
Alternatively, using equation \eqref{eq:E-Proca-def-A0} and
the formula $A_\perp = (A_0-X\interior A)/N$ we can 
treat $L_\Sigma$ as a function of $(g,N,X)$ and $(A,A_0,\dot A, \dot A_0)$.

The energy and momentum densities can be found via the equations
\begin{align*}
\hamE dV_h&= -\frac{1}{2} \frac{\partial L_\Sigma}{\partial N}\\
\momJ dV_h&=  \frac{1}{2} \frac{\partial L_\Sigma}{\partial X}
\end{align*}
taking care in the computation of $\momJ$ to write $L_\Sigma$
using $A_0$ instead of $A_\vdash$. We find
\begin{align}\label{eq:proca-E}
\hamE &= \frac{1}{8\pi}\left[|E|_h^2 + \frac12|dA|_h^2
+ m^2 A_\perp^2 + m^2|A|_h^2\right]\\
\label{eq:proca-J}
\momJ &= \frac{1}{4\pi}\left[\ip<E,\cdot \interior d A>_h -m^2 A_\perp A\right].
\end{align}

Because the additional mass term $\mathcal L_m$
does not involve derivatives of $\secA$, the 
momenta of the system  are the same as for standard
electromagnetism:
\begin{align*}
\Pi_A &= -\frac{1}{2\pi}\ip<E,\cdot>_h dV_h \\
\Pi_{A_\vdash} &= 0.
\end{align*}
The vanishing momentum $\Pi_{A_\vdash}$ implies a constraint:
$L_\Sigma$ must be stationary with respect to $A_\vdash$ and we find
\begin{equation}\label{eq:proca-constraint-v1}
\int_{\Sigma} 
\left[ \left<E,d\, \delta A_\vdash\right>_h + m^2 (A_\vdash/N) \delta A_\vdash\right]\;dV_h = 0
\end{equation}
for any compactly supported function $\delta A_{\vdash}$.  Hence
\[
\div_h E = m^2 A_\perp.
\]
Substituting $m^2 A_\perp = \div_h E$ in the momentum constraint
\eqref{eq:proca-J} we verify by an argument analogous 
to that of equation \eqref{eq:J-as-Pi-integral}
that for any compactly supported vector field $\delta X$,
\begin{align*}
\int_\Sigma  \momJ(\delta X)\,dV_h &=
\frac{1}{4\pi} \int_{\Sigma} \left[\ip<E,\delta X\interior d A>
-(\div_h E) A(\delta X)\right]dV_h\\
&=
\frac{1}{4\pi} \int_{\Sigma} \left[\ip<E,\delta X\interior d A>
+\ip<E, d(A\interior X)>_h\right] dV_h\\
&= -\frac 12 \int_{\Sigma} [\Pi_A(\Lie_{\delta X} A)  + \Pi_{A_\vdash}(\Lie_{\delta X}A_\vdash)].
\end{align*}
Hence the conclusion of Theorem \ref{thm:form-of-J-tensor-edition}
indeed holds for this model.

To see how the constraint \eqref{eq:proca-constraint-v1} is conformally invariant, equation \eqref{eq:proca-constraint-v1} can be rewritten
in terms of the momentum as
\[
\int_{\Sigma} \left[-2\pi \Pi_A(d\delta A_\vdash) + m^2 (A_\vdash/N) \delta A_{\vdash}\;dV_h\right] = 0. 
\]
Let $P_A$ be the unique $n-1$-form such that
$P_A\wedge \eta = \Pi_A(\eta)$ for any 1-form $\eta$; it is easy
to see that $\Pi_A$ and $P_A$ uniquely determine each other independent
of the metric.  Integration by parts then implies
\begin{equation}\label{eq:proca-constraint-v2}
 d P_A + \frac{m^2}{2\pi} A_\vdash \frac{dV_h}{N} = 0.
\end{equation}
At first glance, this constraint involves the spacetime metric 
via $h$ and $N$ and cannot be solved in advance in the context of the conformal method. Remarkably, however, the term $dV_h/N$ already
has a prominent role in the conformal method: it is the 
reference slice energy density $\alpha$ from Section \ref{secsec:ConfMethod}
and is the conformally fixed object that leads to the densitized lapse.
Hence $P_A$ satisfies
$d P_A + m^2 A_\vdash \alpha = 0$,
and solutions of this constraint can be found \textit{a-priori} by the same techniques
as Section \ref{secsec:EMCSF-constraint}.  The densitized lapse of the conformal method is perhaps its least-well understood feature, and 
its natural appearance here for the Proca model may provide an avenue
for clarifying its role.

We finish by deriving the right-hand sides of the LCBY equations.
Rather than write equations \eqref{eq:proca-E}--\eqref{eq:proca-J}
explicitly in terms of $(A,A_\vdash,\Pi_A,\Pi_\vdash)$, 
we follow the procedure outlined at the start of Section
\ref{sec:apps}: we keep the variables appearing
in equations \eqref{eq:proca-E}--\eqref{eq:proca-J} 
but introduce conformal transformation rules
\[
(A,A_\perp,E,h,N)\mapsto (A,A_\perp^*,E^*,h^*=\phi^{q-2}h,N^*=\phi^q N)
\]
that are equivalent to the desired transformation
\[
(A,A_\vdash,\Pi_A,\Pi_\vdash,h,N)\mapsto
(A,A_\vdash,\Pi_A,\Pi_\vdash,h^*,N^*).
\]
We take $(A,A_\perp,E)$ as the matter seed data and 
define $\Pi_A = -1/(2\pi)\ip<E,\cdot>_h$ and $A_\vdash = N A_\perp$.
Fixing $\Pi_A$ and $A_\vdash$ we require
$\Pi_A=-1/(2\pi)\ip<E^*,\cdot>_{h^*}$ and
$A_\vdash = N^* A_\perp^*$ which leads to
\begin{align*}
E^* &= \phi^{-2}E\\
A_\perp^* &= \phi^{-q} A_{\perp}.
\end{align*}
The right-hand sides of the LCBY equations \eqref{eq:Lich-final}--\eqref{eq:LCBY-mom-final} are then
\begin{align*}
16\pi\phi^{q-1}\hamE^* &= 
2\phi^{q-1}\left[|E^*|_{h^*}^2 + \frac12|dA|_{h^*}^2
+ m^2 (A_\perp^*)^2 + m^2|A|_{h^*}^2\right]\\
&=2\left[\phi^{-3}|E|_{h} + \frac12\phi^{3-q}|dA|_{h}^2
+ m^2 \phi^{-q-1}(A_\perp)^2 + m^2\phi|A|_{h}^2\right]\\
-8\pi \phi^q \momJ^* 
&= 
-2\phi^q\left[\ip<E^*,\cdot \interior d A>_{h^*} -m^2 A_\perp^* A\right]\\
&=-2\left[\ip<E,\cdot \interior d A>_{h} -m^2 A_\perp A\right].
\end{align*}
As implied by Theorem \ref{thm:form-of-J-tensor-edition},
$-8\pi \phi^q \momJ^*$ does not involve $\phi$ and
the momentum constraint has decoupled from the Hamiltonian
constraint. The mass term $\mathcal L_m$ in the Lagrangian 
leads to two terms on the right-hand side
of the Hamiltonian constraint that are structurally identical
to those of a scalar field.

The Proca field is treated in \cite{isenberg_extension_1977} in a more complex setting that also
allows torsion, and it arrives at analogous LCBY equations.
That work predates our current understanding of the densitized
lapse as a part of the conformal method (\cite{york_conformal_1999}\cite{pfeiffer_extrinsic_2003}\cite{maxwell_conformal_2014}) and it 
is perhaps curious that the role of the densitized lapse
is not observed there as it is here.
In \cite{isenberg_extension_1977} the constraint \eqref{eq:proca-constraint-v2} 
is written in the form
$dP_A + m^2/(2\pi)A_\perp dV_h$ which is then used to eliminate
$A_\perp$ from the momentum constraint.  Declaring this constraint
to hold by fiat in fact
implies the conformal transformation law
$A_\perp^*=\phi^{-q} A_\perp$ employed here but is arrived at by 
alternative means via an additional assumption (do whatever is
necessary to $A_\perp$ to maintain the constraint while
leaving $P_A$ fixed). 
A novel feature of our work is
the decomposition of Section \ref{sec:decomp} that leads to
$A_\vdash$ rather than $A_\perp$ or $A_0$ as a variable that 
should be kept fixed.  Our approach is parsimonious, inasmuch
as the conformal transformation rule
for $A_\perp$ then arises naturally; we
simply fix $A_\vdash$ and conformally transform the lapse
without any additional ansatzes.

\subsection{Electromagnetism-Charged Dust (EMCD)}
\label{secsec:EMCD}
Our results can be generalized to matter fields
having both tensor-valued and manifold-valued components,
and we present an example of this here,
a fluid consisting of dust with constant charge $\chargec$ per unit particle coupled to electromagnetism. 
A notable feature of our construction is that
the particle density of the solution is directly prescribed and the ratio
of charge to rest mass is spatially constant. 

The Lagrangian is a sum 
$\mathcal L_{\mathrm{EMCD}}=
 \mathcal L_{\mathrm{EM}} + \mathcal L_{\mathrm{Dust}}+
 \mathcal L_{\mathrm{I}}$
of the standard 
Lagrangians for electromagnetism and dust together with a third interaction term:
\begin{align*}
\mathcal L_{\mathrm{EM}} &= -\frac{1}{8\pi} |d\mathcal A|_g^2\;dV_g\\
\mathcal L_{\mathrm{Dust}} &= -2m\mathbf n_0\;dV_g\\
\mathcal L_{\mathrm I} &= 2\chargec \mathcal A(Q)dV_g.
\end{align*}
The fluid variables $\mathbf n_0$ and $Q$ (the rest particle density and the particle flux respectively) are defined in Section \ref{sec:fluids} and $\mathcal A$, as usual, is the electromagnetism 
connection 1-form.
More fundamentally, the field variables consist of $\mathcal A$ and
the fluid configuration map $\Psi:M\to \overline \Sigma$ 
defined at the start of Section \ref{sec:fluids}.

Recall that the particle flux $Q$ is defined by
\[
Q\interior dV_g = \Psi^* \overline \omega
\]
where $\overline \omega$ is the reference particle density on the
material manifold $\overline \Sigma$.  A computation then shows
we can alternatively write
\[
\mathcal L_{I} = 2\chargec\, \mathcal A\wedge \Psi^*\overline \omega,
\]
an expression that does not involve the metric and which 
shows that the charge density of the fluid is proportional
to the number density. Alternatively, 
a separate reference charge density 
$\overline{\omega}_{\chargec}$ could be specified on the material manifold
and the interaction term would be $2\mathcal A\wedge \Psi^*\overline\omega_{\chargec}$.

Because the interaction term does not involve the metric,
it has no impact on the stress-energy tensor, which is
simply a sum of the stress-energy tensors for electromagnetism and dust.
Consequently on a slice $\Sigma_{t_0}$, now identified with $\Sigma$,
\begin{align}
\label{eq:E-EMCD}
\hamE &= \frac{1}{8\pi}\left[|E|_h^2+\frac{1}{2}|dA|_h^2\right] 
+\mathbf n \sqrt{m^2+|P_\Psi|^2}\\
\label{eq:J-EMCD}
\momJ &= \frac{1}{4\pi}\ip<E,\cdot\interior dA> +\mathbf n P_\Psi;
\end{align}
the notation here follows that of Sections \ref{secsec:dust} and 
\ref{secsec:proca}.  In particular, $\mathcal A$ has been decomposed into
$(A,A_\vdash)$, the electric field has the form of equation
\eqref{eq:proca-E}, the material manifold is $\Sigma$ itself,
and the fluid variables $\mathbf n$, which encodes
the initial number density $\omega$, and $P_\psi$,
which represents the fluid momentum density,
are discussed in detail in Section \ref{secsec:conf-fluid}.  As
in Section \ref{secsec:conf-fluid} we have assumed, as we are
free to do, that the initial fluid configuration 
satisfies $\psi(t_0)=\Id$.
 
The interaction term $\mathcal L_I$ does not involve
derivatives of $\mathcal A$ and therefore does not impact
the EM momenta.  We have
\begin{align*}
\Pi_A &= -\frac{1}{2\pi}\ip<E,\cdot>_h dV_h \\
\Pi_{A_\vdash} &= 0.
\end{align*}
On the other hand, $\mathcal L_I$ does contain derivatives of $\Psi$, and this affects the fluid
momentum.  Using equation \eqref{eq:V-vs-dotPsi} we compute
\[
\mathcal A(Q) = \frac{\mathbf n}{N}\left[A_\vdash 
-A(\psi^* \dot\Psi-X)\right]
\]
and therefore
\begin{equation}\label{eq:L_I}
\partial_t\interior \mathcal L_I = 2\chargec\mathbf n( A_\vdash 
-A(\psi^* \dot\Psi-X))dV_h.
\end{equation}
Taking a derivative with respect to $\dot\Psi$
we conclude that $\Pi_\Psi$ picks up an additional term
\[
 -2\chargec\mathbf n\, A\circ\psi^* \,dV_h.
\]
The contribution to $\Pi_\Psi$ from $\mathcal L_{\mathrm{Dust}}$
is the same as in Section \ref{secsec:conf-fluid}
and, recalling that on the initial slice we have assumed $\psi=\Id$,
we conclude
\[
\Pi_{\Psi} = -2(P_\Psi + \chargec A)\, \mathbf n\, dV_h.
\]
As argued in Section \ref{secsec:conf-fluid}, 
$\mathbf n\, dV_h$ is metric independent; it is
the initial number density $\omega$ and therefore
$\Pi_{\Psi}$ can be specified without using 
the metric by $P_\Psi$, $\omega$ and $A$.  

Because the momentum conjugate to $A_\vdash$ vanishes there is
a constraint; for any compactly supported function $\delta A_\vdash$,
\[
\int_{\Sigma} \frac{\partial L_\Sigma}{\partial A_\vdash}[\delta A_{\vdash}] = 0
\]
where $L_\Sigma$ is the slice Lagrangian $\partial_t\interior \mathcal L_{\mathrm{EMCD}}$.  Using equation \eqref{eq:L_I},
a computation analogous to that leading to
\eqref{eq:proca-constraint-v1} implies
\begin{equation}\label{eq:emcd-constraint-v1}
\int_{\Sigma} \left[ \frac{1}{2\pi} \ip<E, d\delta A_\vdash> + 2\chargec\mathbf n \delta A_\vdash\right]\;dV_h
\end{equation}
which implies Gauss' Law
\[
\div_h E = 4\pi \chargec\, \mathbf n.
\]
The constraint \eqref{eq:emcd-constraint-v1} can be alternatively
written
\[
\int_{\Sigma} \left[ -\Pi_\Sigma(d\,\delta A_\vdash)
  + 2\chargec\, \delta A_\vdash \omega
 \right]
\]
which is a metric-free statement and hence a conformally invariant
condition relating $\Pi_A$ and $\omega$.

Although neither Theorem \ref{thm:form-of-J-manifold-edition}
nor Theorem \ref{thm:form-of-J-tensor-edition} apply directly
to this model, we can still verify that a natural generalization
of equations \eqref{eq:J-to-Pi-B-manifold-edition}
and \eqref{eq:J-to-Pi-B-tensor-edition} holds here.
Using equation \eqref{eq:J-EMCD} and the definitions
of $\Pi_A$ and $\Pi_\Psi$ we find that for any
compactly supported vector
field $\delta X$
\[
-2\int_{\Sigma} \momJ(\delta X)\,dV_h = 
\int_{\Sigma} \left[\Pi_A(\delta X\interior dA) + \Pi_\Psi(\delta X)\right]+\int_{\Sigma}  + 2\chargec\int_\Sigma A(\delta X)\mathbf n\, dV_h.
\]
Gauss' Law and integration by parts implies
\begin{align*}
2\chargec\int_\Sigma A(\delta X)\mathbf n\,dV_h &=
\frac{1}{2\pi}\int_\Sigma  (\div_h E) A(\delta X)\, dV_h\\
&= -\frac{1}{2\pi}\int_\Sigma \ip<E,d (\delta X \interior A)> \, dV_h\\
&= \int_\Sigma \Pi_A(d (\delta X \interior A)).
\end{align*}
Finally, these last two equations along with
the vanishing of $\Pi_{A_\vdash}$, the formula $\Lie_{\delta X} = 
d(\delta X\interior A) + \delta X\interior dA$
and the fact that the initial fluid configuration is the identity
yields the desired relation
\[
-2\int_{\Sigma} \momJ(\delta X)\,dV_h = 
\int_{\Sigma} \left[ 
\Pi_A(\Lie_{\delta X} A) + 
\Pi_{A_\vdash}(\Lie_{\delta X} A_\vdash) +
\Pi_\Psi(\Psi_*\delta X)\right].
\]

Because the interaction term does not impact the energy
and momentum densities, the right-hand sides of the LCBY
equations are straightforward combinations of those for
electrovac and dust.  We can take as seed data for
the matter fields $E$, $A$, $\mathbf n$ and $P_\Psi$.
Fixing $\Pi_A$ is equivalent to the transformation rule
$E^*=\phi^{-2}E$ and fixing $\omega$ is implied by the
transformation rule $\mathbf n^* = \phi^{-q}\mathbf n$.
As remarked above, fixing $\Pi_\Psi$ is equivalent
to fixing $\omega$, $P_\Psi$ and additionally $A$.
Thus the variables and transformation rules are identical
to those found in Sections \ref{secsec:conf-fluid},
\ref{secsec:dust} and \ref{secsec:proca} and we find by
the same computation as in those sections that
the right-hand side of the LCBY equations 
\eqref{eq:Lich-final}--\eqref{eq:LCBY-mom-final} are
 \begin{align*}
16\pi\phi^{q-1}\hamE^* &= 
2\phi^{-3}|E|_{h} + \phi^{3-q}|dA|_{h}^2
+16\pi\mathbf n\sqrt{m^2\phi^{-2}+\phi^{-q}|P_\Psi|_h^2}\\
-8\pi \phi^q \momJ^* 
&=-2\ip<E,\cdot \interior d A>_{h} -8\pi\mathbf n P_\Psi
 \end{align*}
The right-hand side of the momentum constraint has decoupled,
and the right-hand side of the Lichnerowicz equation is
monotone decreasing in $\phi$ and can be handled by 
barrier techniques.

The number density is fixed throughout
this construction: it is $\mathbf n\, dV_{h}=\mathbf n^* dV_{h^*}$.
Hence so are the mass and charge densities, which are proportional
to the number density by the constants $m$ and $\chargec$, and the fluid has
a constant charge to mass ratio.  
The
final electric field satisfies $\div_{h^*} E^* = 4\pi\chargec\, n^*$,
and arranging for this along with the  constant
charge to mass ratio requires compatibility between two different
scaling rules.  The standard scaling $E^*=\phi^{-2}E$ for the electric field implies, since $\div_{h^*} (\phi^{-2} E) = \phi^{-q} \div_{h} E$,
that the (scalar) charge density must scale like $\phi^{-q}$.  
But the constant charge to mass ratio then implies that the scalar
mass density must also scale like $\phi^{-q}$, which fortuitously 
agrees with the scaling we used for uncharged fluids 
in Section \ref{secsec:conf-fluid}: $\mathbf n^* = \phi^{-q}\mathbf n$.  If the approaches of \cite{dain_initial_2002} or \cite{isenberg_gluing_2005}  
had been applied to the fluid variables, the scalar charge density would
have scaled like $\phi^{-q}$, but the scalar mass density would not have,
and the resulting solution of the constraints would have had a mass density unrelated to the charge density.  As remarked above, our
construction can alternatively arrange for an 
arbitrary charge to mass ratio as 
a function of space by specifying a separate reference charge density
on the material manifold.

\section*{Acknowledgment}
This work was supported by NSF grants DMS-1263544, DMS-1263431 and 
PHY-1707427. 
We thank Paul Allen and Iva Stavrov for their comments on an early draft of this work,
Volker Schlue for pointing out stiff fluids as a
interesting fluid application, and our colleagues Rafe Mazzeo and Michael
Holst for their partnership with us in the constraint equations
FRG research group that incubated this work.

\bibliographystyle{amsalpha-abbrv}
\bibliography{conformalmatter,ConformalMatter-DAM}
\end{document}